\documentclass[11pt]{amsart}
\usepackage{geometry}                

 \usepackage[all]{xy}

  \usepackage{epsf,epsfig,amsfonts,graphicx,color}

\usepackage{graphicx}
\usepackage{amsfonts,amssymb}
\usepackage{epstopdf}

\def\g{\mathfrak g} 

  \usepackage{url}
   
 \def\R{\mathbb R}
  \def\G{\mathbb G}
\def\bP{\mathbb P}
\def\C{\mathbb C}
 
 \def \V {\mathbb V}

 \def \Z {\mathbb Z}

\newcommand{\Omu}{{\mathcal O}_{\lambda}}

\def\eps{\epsilon}

\def \nbhd { neighborhood }

\def\Ri{ Riemannian}

\newcommand{\beq}{\begin{equation}}
\newcommand{\eeq}{\end{equation}}
\newcommand{\Leq}[1]{\label{#1}\end{equation}}

\newcommand{\dd}[2] { {{\partial #1}   \over {\partial #2}} }

\newtheorem{theorem}{Theorem}
\newtheorem{corollary}{Corollary}

\newtheorem{lemma}{Lemma}
\newtheorem{remark}{Remark}
\newtheorem{case}{Case}
\newtheorem{definition}{Definition}

\DeclareMathOperator{\tr}{\mathrm{tr}}
\DeclareMathOperator{\rank}{\mathrm{rank}}

\title{The $N$-body problem on coadjoint orbits}
\author{Holger Dullin \& Richard Montgomery}
\date{starting date: May 2021; finished: January 2025? }

\begin{document}

\begin{abstract}

 We show (Theorem 3)  that the symplectic reduction of the spatial $n$-body problem at non-zero angular momentum
 is a singular symplectic space consisting of two symplectic  strata, one for spatial motions
 and the other for planar motions.  Each stratum is realized as   coadjoint
 orbit  in the dual of the Lie algebra of the linear symplectic group $Sp(2n-2)$.  The planar
 stratum arises as the frontier upon   taking  the closure of the spatial stratum. 
 We  reduce by   going to center-of-mass
 coordinates  to reduce by translations and boosts 
 and then performing   symplectic reduction with respect to the orthogonal group $O(3)$.  
 The theorem is a special case of a general  theorem (Theorem 2)   which holds for the $n$-body problem
 in any dimension $d$.  This theorem  follows largely from a  ``Poisson reduction'' theorem, 
Theorem 1.    We achieve our reduction  theorems by combining  the   Howe  dual pair perspective of reduction   
 espoued in  \cite{LMS93} with a normal form arising from a  symplectic singular value decomposition due to Xu \cite{Xu}.
 We begin  the paper   by showing how Poisson  reduction by the Galilean group rewrites Newton's
 equations for the $n$-body problem as a Lax pair.  In section 6.4 we show that this Lax pair representation
 of the $n$-body equations is equivalent to the Albouy-Chenciner \cite{AlbouyChenciner98} representation in terms
 of symmetric matrices.  
 
\end{abstract} 

\maketitle

\section{Introduction}  Newton's $n$-body problem  has 
the Galilean group  as a symmetry group.  As a result  we can push  the differential equations
 defining  the   problem  down to  reduced equations on a kind of   quotient space of  $n$-body   phase space
 by the Galilean group.   We   show here that  these reduced equations  form   a  Lax pair:
  \beq
\dot K = [P, K], \qquad K,  P  \in \g = sp(2n-2). 
\Leq{Lax}
and that these in turn are  equivalent to the   Albouy-Chenciner version of the 
 reduced N-body equations described  in their celebrated paper \cite{AlbouyChenciner98}.  
 
 The evolving matrix $K(t) \in \g = sp(2n-2)$ of \eqref{Lax} encodes
 the  ``Galilean shape'' of  a  moving phase point.    The Lie algebra  $\g$  
consists of $2n -2 \times 2n -2$ real square matrices $K$
satisfying $KJ + J K^t = 0$ where $J$ is the constant  $(2n-2) \times (2n-2)$ matrix representing
the standard symplectic form on $\R^{2n-2}$ and described in 
equation   \eqref{defining P} below. 
The  quotient space is a conjugation invariant convex subset of   $\g$.

The   matrix function $P = P(K) \in \g$ occuring in \eqref{Lax}  corresponds to the
Hamiltonian vector field.  $P$ is given by equation \eqref{defining P} below
and   depends on the masses of the bodies and the potential which defines the problem,  assumed to be
 a function of the mutual  distances between bodies alone.


{\it Duality.} It will be essential   to  view the Lax pair equations \eqref{Lax} as being a   Lie-Poisson equations on
the dual space  $\g^*$ of $\g$.
 Being a simple Lie algebra,    $\g$ is isomorphic to   $\g^*$ as a $G = Sp(2n-2)$-space.
 This  isomorphism   is unique up to scale and   we  call it  the Killing isomorphism. 
 The isomorphism   takes any Lax pair equation on $\g$  to  a  Lie-Poisson equation on $\g^*$.
 The former equations automatically have the adjoint orbits as invariant submanifolds
 while the latter set of equations   have the corresponding image coadjoint orbits of $\g^*$
 as invariant  symplectic submanifolds. 
 
 {\it Symplectic Reduction and New Results.}  Our main result, theorem \ref{symp reduction},  relates certain  singular symplectic spaces  consisting of the closure of certain  coadjoint orbits  in $\g^*$ to the (Marsden-Weinstein-Meyer) 
   symplectic reduced spaces for the $n$-body problem, reduced at certain  values of angular momentum.  In particular, we show 
 in theorem \ref{thm: spatial reduced space} that the  symplectic reduction
 of the spatial $n$-body problem with respect to the orthogonal group $O(3)$ at nonzero angular momentum consists of a singular
 space which is the union of two $\g^*$ coadjoint orbits.   We obtain these   results
 by combining the   dual pair perspective of \cite{LMS93}  (sections 4 and 5 therein), 
 with a matrix normal form   due to Xu \cite{Xu}.  We summarize the  dual pair perspective
by  diagram \eqref{structure} below.  We describe the matrix normal form of Xu in Appendix B.

\subsection{Forming the quotient.}
Newton's equations for $n$ bodies moving in Euclidean space $\R^d$ form  a set of $d \times n$ second order differential equations which we   can   rewrite as a system of  first order Hamiltonian equations on 
a phase space $P_{d,n}$ of dimension $2nd$. 
 To get from $(q,p) \in P_{d,n}$ to the matrix  $K$ of \eqref{Lax} 
  follow the diagram down to $G(Z) = Z^t Z $ and set $K = JG$.  
   \beq
\xymatrix{ 
 & P_{d,n}  \ar[d]^{\pi_{trans}}   &  \\
& P_{d,n-1}  \ar[dl]^{L}  \ar[dr] ^{G}& \\
o(d)^* & &  sp(2n-2)^* 
} 
\qquad 
\xymatrix{ 
 & (q,p)  \ar[d]^{\pi_{trans}}   &  \\
& Z = (X, Y)  \ar[dl]^{L}  \ar[dr] ^{G}& \\
 Z J Z^t & &   Z^t Z 
} 
\Leq{structure}
The intermediate object  $Z \in P_{d, n-1}$ is a $d \times 2n-2$ matrix
  representing the phase point $(q,p) \in P_{d, n}$ after quotient by Galilean boosts and translations. 
  We implement this quotient by the usual trick of going into center-of-mass frame and then using Jacobi vectors. 
The  quotient map $\pi_{trans} : P_{d,n} \to P_{d, n-1}$  is a surjective linear map which implements
the  quotient.
See  section \ref{ss: translation reduction} for details.  
  
 \begin{remark}  $P_{d, n-1}$ is not 
 literally a quotient by Galilean boosts and translations
 because the   boosts do not act on phase space in the usual sense.
 Boosts act on Galilean   space-time and their  action  requires
 explicit knowledge of the current time, while the phase spaces
 keep no knowledge of clock time.  Nevertheless, the usual center-of-mass
 trick allows us to non-ambiguously identify $P_{d,n-1}$ with the quotient
 of $P_{d,n}$ by the action of the Galilean boosts and translations of $\R^d$.
 \end{remark}   

The orthogonal group $O(d)$ is the subgroup of the Galilean group which fixes the center of mass $0 \in \R^d$
and  acts on centered phase space $P_{d, n-1}$, with $g \in O(d)$ acting by $Z \mapsto gZ$. 
The Gram matrix
\beq
G = Z^t Z  \in symm(2n-2)
\Leq{Gram}
is invariant under this action, and is a complete invariant. Here   $symm(2n-2)$ denotes the vector space of all     symmetric $(2n-2) \times (2n-2)$  matrices.  
The  image  of the Gram map represents   the full quotient of $P_{d,n}$ by the Galilean group.
 See the end of section \ref{s: dual pairs} and
section \ref{ss: rotation reduction}  for more.

 We need to  explain how  we  think of $G = G(Z)$ as an element of  $\g = sp(2n-2)$
 and of  $\g^* = sp(2n-2)^*$.
Multiply $G$ by the previously mentioned  symplectic structure  $J$  to arrive at 
 \beq
symm(2n-2) \cong sp(2n-2) \qquad;  \qquad G \to JG := K \in \g .   
\Leq{G to K}
The  matrix   $K = J G = JG(\pi_{trans}(q,p))$ is the matrix appearing in equation \eqref{Lax}.
And if $q(t), p(t)$ satisfies Newton's equations then $K(t)$ satisfies \eqref{Lax}.
When we want to view $G$ or $K$ as lying in $\g^*$
use  the Killing isomorphism
$\g \to \g^*$ mentioned earlier.

Matrices of the form \eqref{Gram} are not arbitrary symmetric matrices, rather they  are   positive semi-definite.  
Moreover,   $rank(G) \le d$ since $rank(Z) \le d$.   Write 
$symm_+(2n-2;d)$ for the subset of  $symm(2n-2)$ consisting of  positive semi-definite matrices
  of rank $d$ or less.   Then 
\beq
im(G) = symm_+(2n-2; d)   \subset  symm(2n-2) \cong \g^*
\Leq{symm}
 is  the  image of $G$. 
 
 The Lie group $Sp(2n-2)$ whose Lie algebra is $\g$ acts on $P_{d, n-1}$ by {\it right multipication}
 $Z \mapsto Z T, T \in Sp(2n-2)$.   Essential to what follows is that the Gram map is the momentum map
 for this action after making the identifications $symm(2n-2) \cong sp(2n-2)^*$ just
 described.    As a consequence   $G$ is a Poisson map.
 Putting   this together with some understanding of Poisson maps we have:  
\begin{theorem}
The quotient space  of   $n$-body phase space $P_{d, n}$ by the action of the Galilean
group is the Poisson variety   $symm_+(2n-2;d) \subset sp(2n-2)^* = \g^*$
with corresponding quotient map   $G \circ \pi_{trans}$.
The isomorphism $\g \to \g^*$ takes the   Lax pair
equations \eqref{Lax} to  Lie-Poisson equations on $symm_+(2n-2;d)$,  with Hamiltonian being the Hamiltonian
on centered phase space, understood as a function on the quotient. 
\label{Poiss reduction}
\end{theorem}
 
We move on to discuss   the bottom left arrow of   diagram \eqref{structure},
the map $L(Z) = Z J Z^t$.   
  
\subsection{Angular momentum and spectral invariance}
The Lax pair evolution \eqref{Lax}, like any Lax pair evolution,  preserves
the spectrum of the evolving matrix $K$, since the equation says
that  $K$ evolves by conjugation by some time-dependent symplectic matrix.  What is the physical meaning of $K$'s conserved eigenvalues?
The $n$-body problem admits angular momentum $L$ as a conserved quantity.  In terms of $Z$
we  find that 
\beq
L = Z J Z^t \in o(d). 
\Leq{ang mom}
If the bodies are moving  in Euclidean $d$-space then  $L$ takes values in $o(d)$, the Lie algebra of the orthogonal
group $O(d)$ of $d$-space.    To insure that $L(Z) \in o(d)^*$
as per diagram \eqref{structure}, again  identify $o(d)$ with $o(d)^*$ using that Lie algebra's Killing isomorphism.
   The  spectrum of $K$ agrees with the spectrum of $L$.  This fact was observed
 by Albouy and Chenciner \cite{AlbouyChenciner98} and is reproved here as   lemma  \ref{spectral lemma}  of section \ref{s: ang mom}.

 \subsection{Structure of Paper}  
In section \ref{derivation of eqns} we derive the Lax pair equations \eqref{Lax}.  In section \ref{s: ang mom} we describe
the angular momentum $L$  and   its relation to the Gram map $G$ using   matrix language.  Further on, in section \ref{s: AC = us}
we show that the Lax pair equations \eqref{Lax} are the same equations as those derived by Albouy-Chenciner.
In section \ref{s: dual pairs} we set up the main structure of the paper, that of
dual pairs.   The paper is structured around the fact
that $O(d)$ and $Sp(2n-2)$ form a ``Howe dual pair'' whose  momentum maps are  $L$ and $G$,
a fact observed in \cite{LMS93}.  We define dual pairs in our Poisson context in definition \ref{def: dual pair} below.
This dual pair structure leads almost immediately  to theorem \ref{Poiss reduction}
which  identifies  the Galilean quotient of phase space  as   a Poisson subvariety of   $sp(2n-2)^*$,  specifically  as the subvariety
  of all positive semi-definite matrices $G$.   
We then  identify the {\it symplectic reduced spaces} for the $n$-body problem in $d$-dimensions as closures of coadjoint orbits within this subvariety,
see theorem \ref{symp reduction}.  Astronomers are  most  interested in the case $d=3$
of the spatial three-body problem   but our   constructions hold for any $d$.   The incarnation of our general theorem 2 
for the case   $d=3$ of most interest is  detailed as   theorem \ref{thm: spatial reduced space}.
In section  \ref{s: derivation} we derive the starting point equations \eqref{centered eqns of motion} used in section \ref{derivation of eqns}  to derive the Lax pair. 
In section \ref{s: spatial reduction} we give a full proof of theorem \ref{thm: spatial reduced space}. 
Four  appendices are included.  The first two are included  to make the article more  self-contained. The last two   appendices provide
some partial classification  of the coadjoint orbits which arise inside the subvariety and which seem unavailable elsewhere.

\subsection{History and Antecedents}
 
  Albouy and Chenciner  \cite{AlbouyChenciner98} wrote the  reduced $n$-body  equations as  a first order differential equation
 for a  symmetric   $2n-2 \times 2n-2$ matrix which is our $G$.   They
 realized the quotient by the Galiean group in an identical two-step way as  $G \circ \pi_{trans}$.       
In section \ref{s: AC = us} we  show that these
  Lax pair equations are equivalent to the Albouy-Chenciner equations.

Albouy and Chenciner \cite{AlbouyChenciner98}  describe how Lagrange first reduced the spatial three-body problem using invariant polynomials.
Lagrange's  work could be viewed as  the antecedent of everything above.   Albouy and Chenciner go on to reduce the
general $n$-body problem in $d$-dimensions, achieving many of the results described here along with many others.  
Instead of invoking Jacobi vectors, they  
  use the more symmetric language of ``dispositions'' to reduce by translations and boosts, i.e.{} to
achieve the arrow $\pi_{trans}$ and work in centered phase space.

Dual pairs first appeared in representation theory and quantum mechanics.
  See  Howe   \cite{Howe1} and  \cite{Howe2}.
That dual pairs  arise in connection to the classical  $n$-body problem was pointed out in  
\cite{LMS93}.  That paper lacks the first of the two steps refered to above,  the translation-boost reduction  arrow $\pi_{trans}$ of diagram \eqref{structure}.
Subsequently the papers 
\cite{BolsinovBorisovMamaev99},
\cite{Sadetov02},
\cite{Cerkaski03},
\cite{Dullin13}
 use dual pairs in connection with the $n$-body problems.
 General facts about dual pairs for matrix groups as they arise in geometric mechanics  can be found in \cite{Vizman18}.
In \cite{ADN15} the extension from the $n$-body problem to the regularised $n$-body problem is discussed, which leads to the Lie-Poisson structure $u(m,m)$, along with a Lax pair.

\section{Deriving the Lax Pair}  
\label{derivation of eqns}

 Here we derive the Lax pair equation \eqref{Lax}.
 In section \ref{s: derivation} we show that Newton's equations for the $n$-body problem in the center of mass frame 
 can be written in matrix form as
 \beq
\dot Z = - Z P.   
\Leq{centered eqns of motion}
where $Z \in P_{d,n-1}$ is a  $d \times 2n-2$ matrix coordinatizing  `centered' phase space
and   referred to above several times.     
$P$ is the   $2n-2 \times 2n-2$  infinitesimal symplectic matrix occuring in \eqref{Lax}.      

The derivation  of equation \eqref{centered eqns of motion} requires that   the  potential defining the interbody  forces depends
only on the  interbody distances $r_{ab}$.   When we expand   $Z$ into column vectors
\beq
Z = (X, Y) = (X_1, \ldots, X_{n-1}, Y_1, \ldots, Y_{n-1}). 
\Leq{Z}
 its  first $n-1$ columns $X_1, \ldots, X_{n-1}$ are vectors in $\R^d$   which  represent  
the  positions of the  bodies as encoded by  Jacobi vectors  (see Appendix A) while its last $n-1$  columns $Y_1, \ldots,Y_{n-1}$  represent  the conjugate momentum vectors.
The `small Gram matrix' $b=X^t X $ whose $n \choose 2$  entries are the dot products $X_i \cdot X_j$
contains  precisely the same data as the knowledge of the $r_{ab} ^2$.   The small Gram matrix  forms the top left block of
the big Gram matrix \eqref{Gram}. 

  We also  show in  section \ref{s: derivation} that    
\beq P = JS \text{ with } J =\begin{pmatrix} 0 & I \\ -I &0 
\end{pmatrix} \quad \text{ and } \quad S = \begin{pmatrix} \tilde M^{-1} & 0 \\ 0 & \tilde A  \end{pmatrix}\,,
\Leq{defining P}
where the entries of the symmetric matrix $S$ contain first derivatives of the Hamiltonian.
The $(n-1) \times (n-1)$ symmetric matrix $\tilde A$  is a version of what Albouy-Chenciner call the ``Wintner-Conley matrix''
and depends only on  the small Gram matrix $X^t X $.  The  entries of  $\tilde A$ encode first derivatives of the potential with
respect to the  $b_{ij} =  X_i \cdot X_j$.  
The matrix $J$ is the standard symplectic matrix encoding the symplectic structure on
$\R^{2n-2}$. 
Finally $\tilde M$ is the diagonal matrix of reduced masses (see Appendix A). When using normalized Jacobi vectors $\tilde M = I$.

The game is now to differentiate the big Gram matrix with respect to  time,   using equation \eqref{centered eqns of motion} 
and the relation $K = JG$. 
Since  $J^t = - J$ and $J^2 = -I$ while $S^t  = S$,
we have  $P^t = S^t J^t = - S J$.
From $\dot Z = - ZP$ we   get  
$$\dot Z^t = (-P^t Z^t) =  S J  Z^t $$
It follows that 
\begin{eqnarray*}
\frac{d}{dt} JG&= &   J \dot G \\
 &= & J (\dot Z^t Z + Z^t \dot Z) \\
  &= & J((SJ Z^t ) Z + Z^t(-ZP) ) \\
  & = & JS JG -  J Z^t Z P \\
    & = & PK - K P
    \end{eqnarray*}
 which is our Lax pair equation (\ref{Lax}) with  
 $K = JG $.        
    QED

\section{ Angular Momentum and its spectrum}

\label{s: ang mom}

The angular momentum $L = L(Z)$  for a motion with center of mass zero can
be expressed in terms of the  columns $X_i$ and  $Y_i$
  of $Z$ as  
$L(Z) = \sum_i X_i \wedge Y_i$. 
As  elements of $o(d)$ we have $X_i \wedge Y_i = X_i Y_i ^t - Y_i X_i ^t$.
A simple computation now yields equation \eqref{ang mom} above for angular momentum. 
   Note that
$tr( L^2) = tr(Z J Z^t Z J Z^t) = tr (J Z^t Z J  Z^t Z) = tr(K^2) $ where $K = JG$ as above.
More generally we have: 
\begin{lemma}[Spectral Lemma] 
\label{spectral lemma}
$\tr L^k = \tr K^k$, $k = 0, 1, 2, \ldots$.
Consequently the spectrum of $L$ and $K$ agree except possibly for
the eigenvalue $0$ and its multiplicity. 
\end{lemma}
\begin{proof} 1. 
When expanding out $L^k = (ZJ Z^t) (ZJ Z^t) \ldots (ZJZ^t)$
use $tr(AB) = tr(BA)$ to move the first $Z$ to the end of the line.
Thus
\[
\tr L^k = \tr (ZJZ^t ) (Z J Z^t)\dots (ZJZ^t) 
= \tr JZ^t ZJZ^t \dots ZJZ^tZ
= \tr K^k 
\]
holds for any positive integer $k$.  

For the second part of the lemma, observe that the powers of traces of a square matrix 
determine the coefficients of its characteristic polynomial.  This fact is easily seen by putting
the matrix into Jordan normal form.   Then use Jordan normal form again to observe
that two square matrices of different sizes have the same eigenvalues precisely when
their spectra agree except for possibly the occurrence and multiplicity of zero.
\end{proof}

We give a second proof of the spectral lemma in Appendix B. 

\begin{remark}
The trace of an antisymmetric matrix like  $L$  vanishes for odd $k$.
It follows that   $\tr K^k = 0$ for odd $k$ as well. 
\end{remark} 
\begin{remark} For the two-body problem in space we have
$n =2$ and $d =3$.  The only nontrivial identity expressed in the lemma is sometimes called   Lagrange's (vector) Identity 
$(X \times Y)^2 = |X|^2 |Y|^2 - (X \cdot Y)^2$.
So the theorem can be considered to be a  generalisation of
Lagrange's identity to the  ``$k$th power''  of $2n-2$ vectors in $\R^d$.
\end{remark}

\section{Structure and Dual Pairs}
\label{s: dual pairs}  



The symplectic vector   spaces in the diagram \eqref{structure} are  of the form  
\beq P_{d, m} = \R^d \otimes \R^{2m}.
   \Leq{} 
The factor $\R^d$ represents the  Euclidean vector space in which the bodies move and so comes
with the standard action of $O(d)$.  
The  second    factor $\R^{2m} = \R^m \oplus (\R^m)^*$  is  a symplectic
vector space with symplectic form denoted $\omega_m$  and so comes     with an action of the symplectic group $Sp(2m)$.  
 To construct the symplectic form $\omega$ on  $P_{d,m}$
combine  the  Euclidean structure `$\cdot$' on $\R^d$ and symplectic structure
  on $\R^{2m}$  to define   $\omega(v_1 \otimes w_1,  v_2 \otimes w_2) = (v_1 \cdot v_2)  \omega_m (w_1, w_2)$  
The groups  $O(d)$ and $Sp(2m)$ act on $P_{d,m}$ in a Hamiltonian  manner
and the  actions commute.   

The momentum maps for these actions when $m = n-1$
are our basic maps $L$ and $G$ of \eqref{ang mom} and \eqref{Gram}.
This fact is shown in \cite{LMS93} or can be verified by hand. 
Now $O(d)$ and $Sp(2m)$ acting on $P_{d, m}$ as above form one of the  classic examples of a Howe dual pair.
As a   result their momentum maps $L$ and $G$ form a dual pair in the following 
Poisson sense. 
   \begin{definition} 
   \label{def: dual pair} A pair  $f, g$ of Poisson maps
 \[  \xymatrix{ M
 & P  \ar[r]^{f}   \ar[l]^{g}& N   } \]
   forms a dual pair if $f ^* C^{\infty} (M)$ and $g^* C^{\infty} (N)$
   are each other's commutants within $C^{\infty} (P)$.
   \end{definition}
   Each map  $L$ and $G$ plays two roles: as  the momentum map for  ``its''  group action and as    the  quotient map for the ``other'', or dual  group action.
   $L$,  the angular momentum,  is the momentum map for the action of
  $O(d)$ on centered phase space.  And $L$     implements the quotient map for the action of $Sp(2n-2)$.
 $G$   implements the quotient map for the $O(d)$ action  
  and is the momentum map for the $Sp(2n-2)$ action on centered phase space.
  These facts are verified in section 5 of \cite{LMS93}. 
  
  Not only are $G$ and $L$   invariants for their `dual' groups, 
but they   represent  ``complete  invariants''.  For example, the components of $G$ are
$G_{ij} = Z_i \cdot Z_j$ where $Z_1, \ldots, Z_{2n-2} \in \R^d$ are the column vectors of $Z$.
What Roger Howe has called ``the first main theorem of invariant theory''
 asserts that any polynomial function on $P_{n,d-1}$ 
which is invariant under $O(d)$ can be expressed as a  polynomial in the  
$G_{ij}$.  This implies that $G(Z) = G(Z')$ if and only if $Z' = gZ$
for some $g \in O(d)$.   So the image of $G$ represents the quotient space $P_{n, d-1}/O(d)$ and  $G$
realizes  the quotient projection.
(For a full and careful statetement of this first main theorem of invariant theory  see \cite{Weyl}, theorem (2.9.A)
and the discussions around sections 4 and 5 of  \cite{LMS93}.)

\section{Symplectic Reduction and coadjoint orbits.}
\label{s: symp reduc} 

The   relationship between $O(d)$ and $Sp(2n-2)$ and the dual nature of their momentum maps 
lets us   identify a symplectic reduced space for  one group
with the closure of a coadjoint orbit for  the other.  Here is the prescription  relevant for us.
\begin{lemma} 
The $O(d)$ symplectic reduced space for $P_{d,n-1}$ at $\mu = o(d)^*$
is equal to $G(L^{-1}(\mu))$, viewed as a singular Poisson variety in $\g^*$
\label{reduc lemma}
\end{lemma}
This assertion comprises Theorem 4.4 on p. 198 of \cite{LMS93} where the reader will 
 also find more details of the proof than the one we are about to give.  
 The reader  will also find the claim that this image is  the closure of a single coadjoint orbit lying in
 $sp(2n-2)_+$, a fact we will be later proving by hand.
 
 \begin{proof} 
 We recall the construction of an $O(d)$ symplectic reduced space, also known as a Marsden-Weinstein-Meyer reduced space.
Fix a value $\mu$ for $L$.    Form $L^{-1}(\mu) \subset P_{d, n-1}$.  Let $G_{\mu} \subset O(d)$ denote the isotropy
group for $\mu$, meaning the subgroup of all $g \in O(d)$ such that $g \mu g^t = \mu$. 
Then $G_{\mu}$ acts on $L^{-1}(\mu)$ and, in `very nice' situations $L^{-1}(\mu)$ is a smooth manifold,
and this $G_{\mu}$ action  sweeps out the kernel of the ambient symplectic form restricted to
$L^{-1}(\mu)$ and the quotient space becomes a symplectic manifold:  the symplectic reduced space  $L^{-1}(\mu)/G_{\mu}$.   
Our situation is not generally ``very nice''.  We still define the reduced space to be $L^{-1}(\mu)/G_{\mu}$
but it is typically not a smooth manifold.  Rather, the reduced space is a  singular symplectic manifold, meaning an analytic space
which is stratified in the Whitney sense and whose strata
are smooth symplectic manifolds.    See \cite{LMS93} for details. 

There is an alternative  equivalent construction of the reduced space. Let $\Omu$ be the coadjoint orbit through
$\mu$.  Form $L^{-1}(\Omu) = O(d) L^{-1}(\mu) \supset  L^{-1}(\mu)$.  Now $O(d)$ acts on
$L^{-1}(\Omu)$.  Form the  full quotient $L^{-1}(\Omu)/O(d)$ as a topological space.
Verify that  $L^{-1}(\mu)/G_{\mu} = L^{-1}(\Omu)/O(d)$ as topological spaces.  The singular symplectic structures
also agree if we define them properly on $L^{-1}(\Omu)/O(d)$.     Now the Gram map $G$ 
is invariant under the $O(d)$
action  so that we have that $G(L^{-1} (\mu)) = G( L^{-1}(\Omu))$.  Since
the Gram map implements the quotient by $O(d)$ it follows
that the   restriction of $G$ to $L^{-1}(\mu)$  realizes both  $L^{-1}(\Omu)/O(d)$ and  $L^{-1}(\Omu)/G_{\mu}$.
Since $G$ is a Poisson map this  image, endowed with the   Poisson structure
   it inherits   from $sp(2n-2)^*$,   realizes the $O(d)$-symplectic reduced space at  
  angular momentum $\mu$. 
\end{proof}
 
\subsection{The General Case and normal forms}

 We  will   need  a   normal form  to  label and work with  the   coadjoint orbits
 of interest.  For this purpose   we  
 found it  essential to   use yet a third realization of $\g$: as the space of
 quadratic Hamiltonians on $\R^{2n-2}$.  In this incarnation the Lie bracket
 is the Poisson bracket and the coadjoint action 
 is pullback:  
  $\lambda \mapsto \lambda  \circ g^{-1}$ where   $\lambda \in \g \cong \g^*$  and $g \in Sp(2n-2)$.

Words are in order regarding going back and forth between our three representations of $sp(2n-2)$.
 Given a symmetric $2n-2 \times 2n-2$ matrix $G$ we can form the quadratic Hamiltonian  
$$\lambda_G (x, y) =  \frac{1}{2}  ( x, y ) \cdot G (x, y) ^t    $$
where $(x, y) = (x_1, \ldots, x_{n-1}, y_1, \ldots y_{n-1}) \in \R^{2n-2}$ and 
with the $x_i, y_i$ canonically conjugate coordinates on $\R^{2n-2}$.
Thus, for example, if $G_{11} = 1$ while $G_{ij} = 0, i, j \ne 1$ then
$\lambda_G (x,y) = \frac{1}{2} x_1 ^2$.     To return to $G$, given the
quadratic Hamiltonian $\lambda$ we set $G = d^2 \lambda$, the Hessian of $\lambda$.
To form $K = JG$ from $\lambda$ take the Hamiltonian vector field of $\lambda$
so that $K(v) = (J d \lambda ) (v)$ for $v \in \R^{2n-2}$.  In order to reverse and go   from $K$ back to $\lambda$
set $\lambda (v) = \frac{1}{2} \omega (v, Kv)$ where $\omega = \sum dx_i \wedge dy_i$ is
the symplectic form.  

For the rest of this section we use the quadratic Hamiltonian realization of points of $sp(2n-2)$.
We write $\lambda$ for a typical quadratic Hamiltonian.  
We think of $\lambda \in \g^*$ with $\g = sp(2m)$, having 
set
$$m = n-1.$$
The  coadjoint action is by pullback:  $\lambda \mapsto \lambda \circ g$.
Using this action we can bring  any   $\lambda \in sp(2m)^*_+$
 into the normal form:
 \beq
   \lambda =\frac{1}{2}  \sum_{j=1}^p  \omega_j^2 (x_j^2 + y_j^2) + \frac{ 1 }{2} \sum_{p < j \le p+q }  y_j ^2 
 \Leq{normal form},
where the $\omega_j  \ne 0$ and $p+q \le m$.
This fact is proved in Appendix \ref{s: symp norm form}.
\begin{definition} 
We   call the integers $(p,q)$ appearing in the sums
of the normal form \eqref{normal form} the {\it numerical invariants} of $\lambda$.  
We call the $\omega_j ^2$, $j =1, \ldots, p$,
its spectral invariants.  (We insist $\omega_j \ne 0$.)  We call   $2p +q$
the rank of $\lambda$, or sometimes the `rank of the motion space' for $\lambda$. 
\label{def: invariants}
\end{definition}
\noindent The integer $2p$ is the dimension of the largest symplectic subspace of
$\R^{2m}$ on which $\lambda$ is positive definite.  The nonzero numbers 
$\pm i \omega_j ^2$ are the nonzero eigenvalues of  $K \in \g$, the linear symplectic
map corresponding to $\lambda$.     The  spectral invariants   are also the 
nonzero eigenvalues of any angular momentum $L$ coming from the same $Z$,
see lemma  \ref{spectral lemma}.   The integer $q$ corresponds to  the number of  Jordan blocks
for  the generalized eigenvalue $0$ for $K$.  The rank $2p+q$  is the rank of $Z$ and equals
the dimension of the affine space swept out by the Newtonian solution
for  any initial condition $(q, p) \in P_{d, n}$ which has
has normal form $\lambda$ when this affine space is viewed from the center-of-mass frame . 

\begin{lemma} The coadjoint orbit through $\lambda \in  symm(2n-2)$
is  uniquely determined by the numerical invariants
and spectral invariants of $\mu$.  (See definition \ref{def: invariants}.)  
\label{lem: invariants}
\end{lemma}

 The first term of the
normal form \eqref{normal form}  corresponds to the semi-simple part of $K$ and the second term (if present) 
to its nilpotent part.  We have that  
\beq
rank(G) = rank(K) = 2p +q
\Leq{D =dim}
while
\beq
rank(L) =  2p
\Leq{}

\begin{corollary}  Let $(p,q)$ be the numerical invariants of a coadjoint orbit
in $sp(2m)_+$ and $\omega_j^2, j=1, \ldots, p$ its spectral invariants.  Then: 
\begin{itemize}
\item[(i)]  $q = 0 \iff $ the orbit is closed
\item[(ii)]  $ p=0 \iff L = 0$
\item[(iii)]  $d \le 3 \implies p \le 1$ and  if $p=1$ then  $\omega_1 ^2 = \| L \|^2$.  
\end{itemize}
\end{corollary}

\begin{proof}
Here is the proof of   (i).  It is well-known that a (co)adjoint orbit through a semi-simple element of
a   matrix Lie group is closed.  Since $q= 0$ corresponds to the normal form being
semi-simple we get $q=0 \implies$ the orbit is closed.  We could use
the same type of argument to show $q > 0 \implies $ not closed, but we
prefer to verify this directly.  Look at the normal form.  Consider the linear symplectic transformation $x_j \to x_j /\eps_j,  y_j \to \eps y_j$
for $j = p+1, \ldots, p+q$, while $x_j \to x_j, y_j\to y_j$ for $j$ outside of  this range.
This transformation takes $h$ to $\frac{1}{2}  \sum_{j=1}^p  \omega_j^2 (x_j^2 + y_j^2) + \frac{ 1 }{2} \sum_{j=p+1}^{p+q} \eps_j ^2 y_j ^2$.
Now let any one of the $\eps_j \to 0$ to see that all the orbits which have the same spectral invariants
but numerical invariants $(p, q')$ with $q' < q$ are in the closure of the  orbit through $\lambda$.

To prove (ii) use the spectral lemma, lemma \ref{spectral lemma}.  

\end{proof}

\vskip .3cm

Specifying the   numerical   and   spectral invariants  
uniquely specifies a coadjoint orbit in $sp(2m)_+$.   The orbit is closed if and only if $q =0$.
In all cases the closure of the coadjoint orbit realizes a  symplectic
reduced space for the $m+1$-body problem in $D$ dimensions.  We reduce
at any  value of angular momentum
$L$  which shares the specified spectral invariants.  
 
 \begin{theorem} 
 \label{symp reduction} Suppose that a coadjoint orbit in $sp(2m)^* _+$ has numerical invariants
 $(p, q)$ and  spectral invariants $\omega_j ^2, j=1, \ldots, p$. 
 Write   $d = 2p +q$ for the rank of the points in this orbit.  Let $L \in o(d)$ be an angular momentum sharing its spectral invariants.  
   If  $q = 0$ then this coadjoint orbit is closed within $sp(2m)^*$ and is realized by
  the   symplectic reduction of the $m+1$-body problem in $\R^d$
 reduced at $L$.  
 Otherwise the orbit's  closure consists of  the union of  $q+1$ coadjoint orbits, namely those
 orbits having the  same semi-simple part but
 smaller Jordan blocks  and so numerical invariants $(p, q')$ with  $q' =0, 1, \ldots, q$,
 the choice $q' = q$ corresponding to the original orbit.  The orbit's  closure 
 realizes the symplectic reduction of the $m+1$-body problem in $\R^d$ reduced at $L$.
  The   $(p, q')$ type orbits with $q' < q$
 correspond to the projections to the   reduced space of phase points $Z \in P_{d, m}$
 with $rank(Z)  = 2p +q' < d$.  
 \end{theorem}
 
 \begin{proof}
 With the exception of the description of the closure of the orbit, 
 we  already  proved   the theorem by proving lemma \ref{reduc lemma}. 
 See also   Theorem 4.4 of the earlier  \cite{LMS93} which is a very similar assertion.  
 (See Theorem 5.1 of \cite{LMS93} for the case of angular momentum zero.) 
 For the orbit closure business,  consider the symplectic scaling $x_j \to x_j/\eps_j, y_j \to \eps_j y_j$
 for $p < j \le p+1$ which converts $\lambda$ to 
 $\frac{1}{2}  \sum_{j=1}^p  \omega_j^2 (x_j^2 + y_j^2) + \frac{ 1 }{2} \sum_{p < j \le p+q } \eps_j ^2 y_j ^2$.
 Letting the $\eps_j \to 0$ appropriately we obtain the normal forms
 with numerical invariants $(p, q'), 0 \le q' < q$
 with the same spectral invariants. We cannot change the spectral invariants by taking closure,
 and so have accounted  for all normal forms of all orbits in the closure by varying $q'$ between $0$ and $q$.   
 \end{proof}

\subsection{The case of bodies in 3-space}

   Here we  work out and record    details 
   of   theorem \ref{symp reduction}  when $d=3$ and $L\ne 0$, that is,
   when the bodies
   move in standard three-space with non-zero angular momentum.   In this case the rank $d$
   of the motion space is given by  $d = 2p+q =3$. Since $L \ne 0$
   we must have $p> 0$ which leaves the only possibility   $p =1, q=1$.  
   The   normal form \eqref{normal form} becomes   
   \beq
   \lambda_{\omega}: = \frac{1}{2} \omega^2 (x_1 ^2 + y_1^2)  + \frac{1}{2} y_2 ^2; \qquad \omega^2 = \|L \|^2 \ne 0.
   \Leq{quad norm form} 
   The symplectic transformation $(x_2, y_2) \to (x_2/  \eps,  \eps y_2)$ while leaving 
   $(x_1, y_1)$ unchanged takes $\lambda$ 
   to $\frac{1}{2} \omega^2 (x_1 ^2 + y_1^2)  + \eps^2 \frac{1}{2} y_2 ^2$.
   Let  
   $\eps \to 0$ to establish  that
   the   closure of the orbit through $\lambda$ contains the
   orbit through 
   \beq
   \hat \lambda_{\omega}: = \frac{1}{2} \omega^2 (x_1 ^2 + y_1^2)  
   \Leq{quad planar norm form} 
   corresponding to $q =0$ and $p =1$. 
   The matrices in this later orbit all have rank $2$
   so represent planar initial conditions.  They have the  same spectral parameter $\| L \|^2$
   and hence the same angular momentum.

   \begin{theorem} 
   \label{thm: spatial reduced space} Consider the      $O(3)$ symplectic reduced space for the spatial 
   $n$-body problem at a {\bf non-zero}  angular momentum  $L= \mu \in o(3)$.  The isomorphism
   described by lemma \ref{reduc lemma}  yields a diffeomorphism between  this reduced space and  the   closure of the $Sp(2n-2)$- coadjoint orbit through   
   the normal form $\lambda_{\omega}$  of \eqref{quad norm form}  with   $\|\mu \|^2 = \omega^2$.    This  closure consists of two   orbits, the
   original orbit whose points have rank 3,  and an  added orbit labelled by the planar normal form $\hat \lambda_{\omega}$ of equation 
   \eqref{quad planar norm form}  whose points have rank 2. 
   This added orbit forms the singular locus of the   reduced space  and 
   realizes  the   $O(2)$-symplectic reduction
   of the planar $n$-body problem at this same angular momentum.   A \nbhd of such a planar  point inside the spatial reduced space  
   is    diffeomorphic to  $Cone(\R P^{\ell}) \times \R^s$ where   $\ell+1 = 2(n-2)$ is the codimension of the locus of planar points
   within $L^{-1}(\mu)$   and where  
     $s + 2(n-2) = 6n -10$ is the dimension of the spatial reduced space.  
   \end{theorem}

\begin{figure}[h]
\scalebox{0.4}{\includegraphics{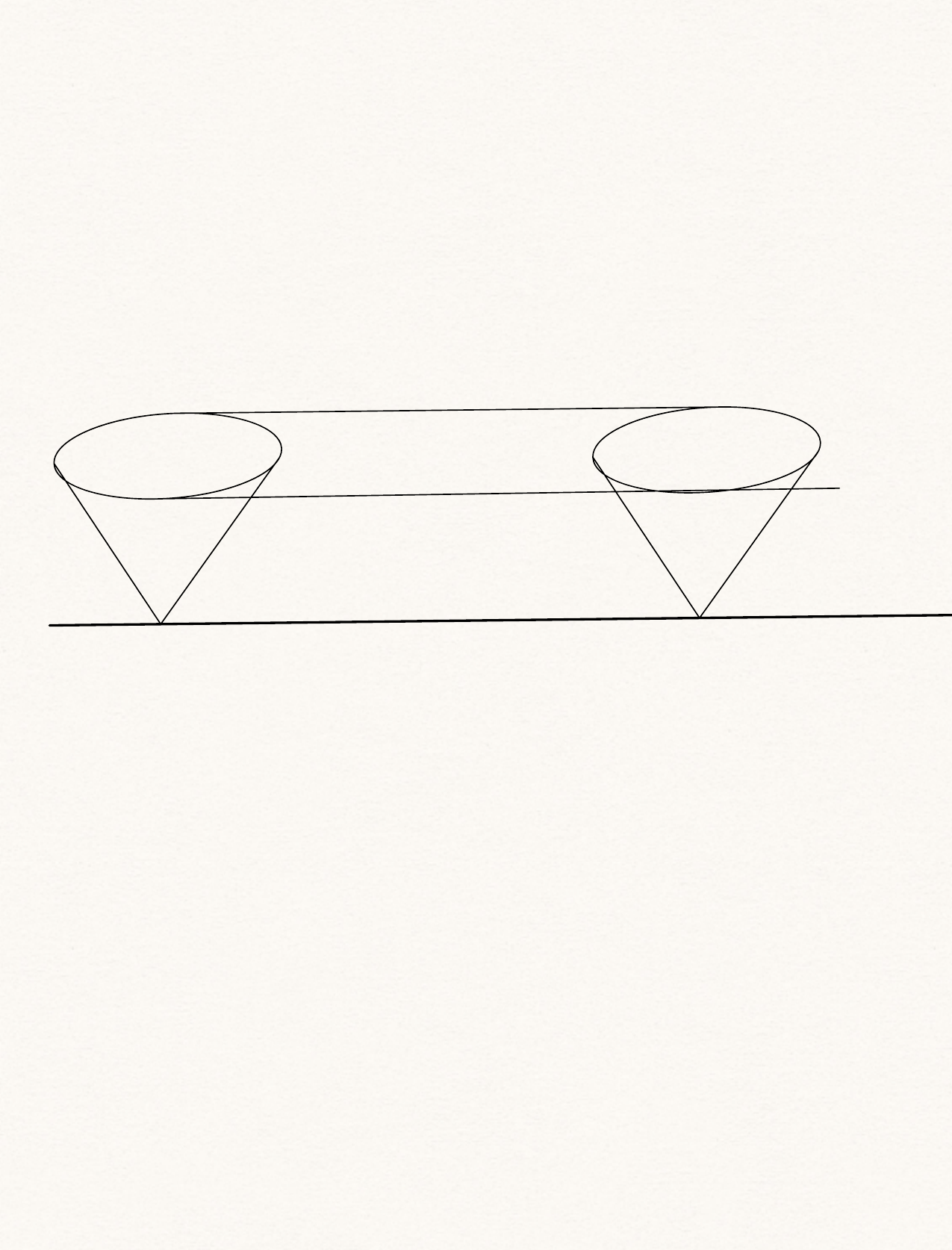}}
\caption{The planar locus is the singular locus for the spatial reduced spaces at non-zero angular momentum.
The local structure around the planar locus is that of smooth manifold times a cone over a real projective space.}  \label{fig: cone0}
\end{figure}
   
   The space $Cone(\R P^{\ell})$ in the last sentence of the  theorem  denotes  the cone  over real projective space of   dimension $\ell$.
   This cone arises as the quotient space $(\R^{\ell + 1})/ \pm 1$ where
   the `$\pm 1$' means we    identify $v$ with  $-v \in  \R^{\ell + 1}$.
   
    \begin{proof} Except for the last two sentences regarding the singularities of the reduced space,
    the theorem is an immediate corollary of the previous theorem.  For proofs of these last
    two sentences and a more physically oriented perspective on the theorem, see the proof which takes up 
    section \ref{s: spatial reduction}. 
    \end{proof} 
    
    \begin{remark} 
    That the planar and spatial strata are different coadjoint orbits has a dynamical consequence.
    No spatial (i.e. rank 3) solution can, in finite time, evolve into a planar (rank 2) phase point.
  \end{remark}

\begin{remark} The reduced space for angular momentum zero  was described   in \cite{LMS93}. That reduced space
   consists of 4 coadjoint orbits whose points have rank   0, 1, 2 or 3.   Normal forms 
 are $h =0$, $h = \frac{1}{2} y_1 ^2$, $h = \frac{1}{2} (y_1 ^2 + y_2 ^2 )$ and $h = \frac{1}{2} (y_1 ^2 + y_2 ^2 + y_3 ^2 )$.
  \end{remark}

\subsection{Resolving the cone singularity: $SO(3)$ reduction}  
We can resolve the singularity along the planar locus  just described in theorem \ref{thm: spatial reduced space}
if we  reduce by  the 
proper orthogonal group $SO(3) \subset O(3)$ instead of the full orthogonal group $O(3)$.
\begin{theorem} 
\label{so3 reduced space}Let $\mu \ne 0$ be a non-zero value of angular momentum.
The  symplectic reduced space for the $SO(3)$ action on   $P_{3, n-1}$,
reduced at $\mu$,  is smooth and  resolves the singularity of the $O(3)$ reduced
space, reduced at this same $\mu$, as   described in theorem \ref{thm: spatial reduced space}.
The fact that $SO(3) \subset O(3)$ and $O(3)/SO(3) = O(2)/SO(2) =  \Z_2$
induces a  canonical Poisson map to
the singular  symplectic $O(3)$ reduced space
which is a   2:1 branched cover branched over 
the singular locus of the $O(3)$ reduced space.  
\end{theorem} 

{\sc sketch of proof.} 

We may take  $\mu$   aligned with the $z$-axis.  Then  
$O(3)_{\mu} = O(2)$ while $SO(3)_{\mu} = SO(2)$, where the `2' refers to transformations which map the  $xy$-plane
to itself and the subscripts $\mu$ indicate  isotropy subgroups  (relative to the coadjoint action) of $\mu$.  
It follows that  the  $SO(3)$ reduced space is $L^{-1}(\mu)/SO(2)$
while the $O(3)$ reduced space is  $L^{-1}(\mu)/O(2)$.  
The key facts are that $\mu$ is a regular value of $L$, 
that $SO(2)$ acts freely on $L^{-1} (\mu)$, 
and that  $O(2)$ does not.   This
last fact arises  from the presence of the reflection   $(x, y,z) \mapsto (x, y, -z)$ which is in $O(2)$ but not in $SO(2)$.
The points of  $L^{-1}(\mu)$ with extra symmetry
are the planar points,  which are those whose 3rd row is all zero.  The planar points    comprise
$L^{-1}(\mu) \cap P_{2, n-1}$.  Their   isotropy group is the 
 $\Z_2$ generated by the reflection. 
Non-planar points of $L^{-1}(\mu)$ have trivial isotropy relative to the
$O(2)$ action on centered phase space.   We have added two remarks  in 
 section \ref{s: spatial reduction}, remark  \ref{1 on so3} and remark \ref{2 for so3},  which detail
the other small changes needed to go from the proof of theorem \ref{thm: spatial reduced space} 
to   the proof of theorem \ref{so3 reduced space}.

\subsection{Extraneous dimensions and dimension counts.} 
 Most of the  evolution described by the  Lax pair representions (\ref{Lax}) is irrelevant
  when  our primary   interest is the $n$-body problem in 3-space.
  The Lax pair evolution space,  $sp(2n-2)$
  has  dimension $2n^2 - 3n +1$  which grows    quadratically with  $n$.  The  dimension of the  phase space for  n bodies
 moving in 3-space is $6n$.  After reduction by the Galilean group its image in $sp(2n-2)$,
 namely the space of   positive semi-definite matrices in $sp(2n-2)$ whose rank is $3$
 or less  has dimension $6n-9$.  The rank 3 matrices are foliated by symplectic leaves,
 these being labelled by the length $\|L \|$ of the angular momentum.
 Each leaf is a symplectic manifold of dimension  $6n-10$ which
 we can identify with the rank 3 subset of the spatial $n$-body problem's symplectic reduced space 
 for any value of $L$ with this fixed   value of $\|L\|$, provided   $\|L \| \ne 0$.   
 (When $L = 0$ the dimension of the symplectic reduced space drops further
 to $6n-12$.)  For $n$ large only the thinnest shell within the matrix space is relevant for the 
 spatial $n$-body problem!

\section{Deriving the  Quotient Map and Reduced Equations}
\label{s: derivation}
\subsection{ Newton's equations}
\label{ss: Newton's equations}
Start with Newton's equations 
$$m_a \ddot q_a = \nabla_a U (q_1, \ldots, q_n), \quad a = 1, \ldots, n .$$
for  $n$ bodies moving in  $d$-space under the influence of a potential.  
The   indices  $a= 1, \ldots, n$ label the $n$   bodies. The $q_a \in \R^d$
are the instantaneous positions of these bodies and the   $m_a > 0$   are their  masses.
 $U = U(q_1, \ldots, q_n)$ is the negative of the potential defining the dynamics
 and depends on   positions  only through   mutual
 distances 
 $$r_{ab}  = |q_a -q_b|.$$
    The symbol ``$\nabla_a$'' denotes the gradient with respect 
    to $q_a$,  in other words  the gradient we get by   allowing $q_a$ to vary while all the other
 $q_b$, $b \ne a$, are fixed.  
 
Rewrite Newton's equations in   matrix form
 \beq
 \ddot q M =  q A, \quad \text{ with }   A = A(r_{ab}) .
 \Leq{N2}
The evolving configuration $q$ is a $d \times n$ real matrix whose columns are   the  position vectors $q_a$ listed in order.   $M$ and $A$ are $n \times n$ matrices with
 $$ M = diag(m_1, \ldots, m_n).$$ 
constant and $A = A(r_{ab})$  a   symmetric
  matrix-valued function of  the mutual distances   defined by  
  \beq
\nabla U (q)  =  - q A.
\Leq{WC1}
Formula \eqref{WC1}   holds by virtue of the chain rule and the fact that $U$ depends only on the mutual distances.
  $\nabla U$ is the $d \times n$ matrix with columns the $\nabla_a U$. 
To derive \eqref{WC1} write
$U(q) = W(r_{ab}^2)$ where $W$ is a function of the symmetric matrix of squared distances $r_{ab}^2 = |q_a - q_b|^2$.
Use $\nabla_ a (r_{ab}^2) = 2 (q_a -q_b)$ and the chain rule $\nabla_a U = \sum_b \dd{W}{r_{ab}^2} \nabla_a (r_{ab}^2)$
and regroup to get $A$. Thus $A_{ab} = -2 \partial W/\partial r_{ab}^2$ for $a \not = 0$ and the diagonal elements 
of $A$ can then be obtained from the (translation symmetry requirement!) that the
 row and column  sums of $A$ must vanish.  For $k, \ell$ positive integers we will  identify $\R^k \otimes \R^{\ell}$ with
$k \times \ell$ real matrices so that our evolution occurs in 
  $$q = (q_1, \ldots, q_n) \in \R^d \otimes \R^n$$

\begin{remark}
 Equation \eqref{WC1} can be rewritten $\ddot q = -q A M^{-1}$ which
 is the form of Newton's equation used in   \cite{AlbouyChenciner98} where   $AM^{-1}$ is called  the ``Wintner-Conley matrix'' (up to a factor $-2$). 
  Why ``Wintner-Conley''?  $AM^{-1}$  was originally used in applications regarding  central configurations.
 Pacella  attributed   it   to Conley in an  article  \cite{Pacella}  which was written soon after Conley's death.  
 See also Albouy \cite{Albouy}, p. 477.   
 Wintner wrote down this same  matrix in \cite[{\S}356]{Wintner41}.
\end{remark}

\subsection{Centering phase space}
\label{ss: translation reduction}  

We use the introductory physics argument to  form 
  the quotient space  of our n-body phase space $P_{d,n}$  
by boosts and translations.     Using a  boost followed by  a translation
we may  transform any given   solution
to the $n$-body problem to a new solution  whose linear momentum 
and  center of mass are both zero, thus effectively reducing the dimension
of phase space by $2d$ dimensions.  Our  new solution then satisfies the equations
$\sum m_a \dot q_a = 0$ and $ \sum m_a q_a = 0$   throughout its evolution. These   two vector equations  define
a  linear subspace of
$P_{d,n}$  which we will call
{\it centered phase space} and which is the claimed  quotient.

 Jacobi vectors are a list of $n-1$  vectors in $\R^d$ which provide us with a   linear isomorphism
between   centered phase space and $P_{d,n-1}$. 
To explain the process,  use velocities
instead of  momenta to describe points of phase space, so 
that centered phase space  equals  $C \oplus C$
where 
$$C =  \{q: \R^d \otimes \R^n:  \sum m_a q_a = 0 \} = \{v \in \R^d \otimes \R^n:  \sum m_a v_a = 0 \}$$
 is {\it centered  configuration space}. Jacobi vectors provide a linear isomorphism
between each copy of  $C$ and $\R^d \otimes \R^{n-1}$.

To form the Jacobi vectors,    observe that the   column vectors $q_a \in \R^d$ 
of $q$ are
determined by   the matrix identity  
$$q = \sum q_a e_a$$ where $e_1, e_2, \ldots , e_n$ is the
standard basis of our label space $\R^n$, written out as   row vectors.
Thus  $e_1 = (1, 0, \ldots, 0), e_2 =  (0, 1, 0, \ldots ,0), \ldots , e_n = (0, 0, \ldots, 1)$
and each $q_a e_a$ is an  $n \times d$ matrix whose rank is at most one.  
Select a different basis  $E_1, \ldots , E_n$ for label space
and expand $q$ out in this new basis: 
\beq
q = \sum_{a =1} ^n Q_a E_a,  Q_a \in \R^d.
\Leq{Jac expansion}
Suppose that the new basis is {\it mass metric-orthonormal} and that its  final vector takes  the form  
$$E_n =  \lambda (1, \ldots, 1).$$ Then the  first $n-1$ vectors $Q_1, \ldots , Q_{n-1}$ are  Jacobi vectors.
(The mass- metric  on label space is the inner product defined by $\langle e_a , e_b \rangle = m_a \delta_{a b}$.)
 The significance  of this   final basis vector $E_n$  is that its direction
 $\hat E_n =   (1, 1 , \ldots, 1)$ generates the action of the translation group of 
$\R^d$  on configuration space  in the sense that
to    translate a configuration $q$ by 
$\tau \in \R^d$ we perform the matrix addition $q \mapsto q + \tau \hat E_n$. 
 \footnote{  It is often conceptually helpful to form the translation space 
$$\mathcal T : =\{ v E_n:    v \in \R^d \} \subset \R^d \otimes \R^n $$
and to realize  that $$C = {\mathcal T}^{\perp}$$
where the perpindicular is with respect to the mass inner product on $\R^d \otimes \R^n$. }

A computation shows that with the  $E_a$ as above the last   vector   
$$Q_n = k \sum m_a q_a, k \ne 0$$
of expansion  \eqref{Jac expansion} is  a multiple (namely $k \sum m_a$) of the center of mass vector.
Also, since $E_a \perp E_n$, $a <n$ we have that  any vector 
$c \in C$ can be uniquely expanded as  $c = \sum_{a = 1} ^{n-1} Q_a E_a$.
The $(Q_1, \ldots , Q_{n-1})$ are the  Jacobi vectors. Put  them together
as column vectors of a matrix we get the promised  $d \times n-1$ matrix $[Q_1, \ldots, Q_{n-1}] \in \R^d \otimes \R^{n-1}$ and consequent linear  isomorphism $C \cong \R^d \otimes \R^{n-1}$. 
A fundamental property of the Jacobi vectors is that they continue to diagonalize the kinetic energy:
 $$K =  \frac{1}{2} \sum_{a < n}  \mu_a |\dot Q_a |^2 + \mu_n |\dot Q_n |^2, \mu_i > 0$$
where, recall, that  $\dot Q_n = 0$ if our curve $q(t)$ lies   in $C$. 

Arrange  the new basis vectors $E_a$  in order to form   the rows of the $n \times n$ matrix $T$ whose   last row is $E_n$.  Then the  previous representation, \eqref{Jac expansion},  of $q$ can
 be re-expressed as 
  $$q = Q T$$ from which it follows that  
Newton's equations in the new linear coordinates $Q \in \R^d \otimes \R^n$ become, after right multiplication by $T^t$ :
\[
     \ddot Q T M T^t = Q T A T^t \,.
\]
The mass orthogonality of the rows $E_a$ of $T$  means that $T M T^t $ is diagonal. 
The diagonal entries of the diagonal matrix $T M T^t$ are the ``reduced masses'' $\mu_i$ in the above expansion of $K$.
Omitting the   last row and column  of $T M T^t$ gives a matrix $\tilde M$.  Omitting the last row and column of $TAT^t$ gives a matrix $\tilde A$, which is a translation reduced Wintner-Conley matrix which is again symmetric. This is the matrix $\tilde A$ that appears in the Lax equations. Renaming  the first $n-1$ columns of $Q$ by $X = (X_1, \ldots, X_{n-1})$, the translation reduced Newton's equations are
\[
    \ddot X \tilde M = X \tilde A \,.
\]
The corresponding translation reduced Hamiltonian equations are
\beq
 \dot X = Y \tilde M^{-1},  \qquad \dot Y = - X \tilde A,
 \Leq{centered Newtons}
 where $\tilde A$ is an $(n-1) \times (n-1)$ symmetric matrix function
 of the dot products $b_{ij} = X_i \cdot X_j$. 
 
 Put  $X$ and $Y$  together into the  single $d \times 2n-2$ matrix as per \eqref{Z}: 
   $ Z = (X, Y) =  (X_1, \dots, X_{n-1}, Y_1, \dots, Y_{n-1}) \in P_{d, n-1}$.
  Written in terms of   $Z$,  Newton's equations now  take  the 
claimed form  of equation \eqref{centered eqns of motion}:  
$
\dot Z = - Z P
$
where
\beq
P =\begin{pmatrix} 0 & - \tilde A (b) \\  \tilde M^{-1}  & 0 \end{pmatrix} , \quad   b =  X^t X \,,
\Leq{Ham2}
with $P$ as defined in equation \eqref{defining P}. 


  \begin{remark} Words are in order regarding the relation 
  between $A$ and $\tilde A$.  We had
  obtained  $A$ by writing the   potential 
 $U$ as a  function  of the   squared mutual distances  and then applying the chain rule to get $\nabla U (q) = -q A$ with $A$ a function of the  squared mutual distances  $r_{ab}^2$. 
Both sets of   functions   $r_{ab} ^2$ and $X_i \cdot X_j$
  are quadratic  functions on configuration space  $E_{d,n}$  which   are invariant under the action of the   isometry group   of $\R^d$.  
  The space of all real-valued quadratic polynomials on 
  $\R^d \otimes \R^n$ which are   invariant under the action of the isometry group of $\R^d$ 
  forms a real vector space of   dimension $n \choose 2$. The squared distances $r_{ab}^2$ form one basis,  and
the entries  $X_i \cdot X_j$ of the small Gram matrix $b= X^t X$ form another.      
In particular each $r_{ab}^2$ can be written as a linear combination
of the $X_i \cdot X_j$ so that we can write the negative of the potential $U$ as
a function of $b$:   $ U(X) = \tilde W (b)$
for some function $\tilde W$ of $b$. 
An application of the  chain rule   yields  
\beq
\nabla U (X) = -X \tilde A (b),  \quad b = X^t X
\Leq{Ham1} 
where the $(n-1) \times (n-1)$ symmetric-matrix valued function  $\tilde A = \tilde A (b)$ is linear in $d \tilde W$.
Explicitly  $\tilde A_{ij} = - \dd{\tilde W}{b_{ij}}$  for $i \ne j$ while $\tilde A_{ii} = - 2\dd{\tilde W}{b_{ii}}$. 
{\it By slight abuse of notation we will also refer to $\tilde A$ as the Wintner-Conley matrix.}
The translation-invariance of $U$ is what implies that $\hat E_n A = 0$
and allows for us to cut $A$ down to the matrix $\tilde A$ with one less row and column than $A$.

Regarding the equivalence between the forms $\tilde A$ and $A$ see also section 4
of \cite{AlbouyChenciner98} where  they look at what it means for  two matrices
to represent the same bilinear form on the dual space ${\mathcal D}^*$  to their space of dispositions.
See also Moeckel \cite{Moeckel}. 

\end{remark}

\begin{remark}
When we choose the   basis $E_a$ for label space to be not just  orthogonal but  orthonormal then we
 the resulting Jacobi vectors ``normalized Jacobi vectors''.
  For such a choice we find that  $\tilde M = Id$ and that  the    kinetic energy is $\tfrac12 \sum |\dot Q_a|^2$.
  \end{remark}

\subsection{Reducing by Rotations.} 
\label{ss: rotation reduction}   
In the last section we arrived at   $P_{d,n-1}$  as a realization
of the quotient of   phase space by  the action of translations and boosts.  
Rotations and reflections of $d$-space remain and act on   $P_{d,n-1}$ by
$Z \mapsto g Z$ where $g \in O(d)$ represents a rotation or reflection
about the center of mass.   

The  Gram matrix
$ G = Z^t Z$ 
is invariant under this $O(d)$ action and is  a
``complete invariant'' according to the first main theorem of invariant theory.  See the book by Weyl, \cite{Weyl},
Theorem 2.11.A, p. 64 for the statement of this theorem in our case.
Since the quotient map $G$ is also the momentum map for the
$Sp(2n-2)$ action, it follows that   the quotient space $P_{d, n-1} /O(d)$ is isomorphic,
as a Poisson manifold,  to the image of $G$ within $sp(2n-2)^*$.

\subsection{Comparisons with Albouy-Chenciner}
\label{s: AC = us}

In their  influential paper \cite{AlbouyChenciner98}
Albouy and Chenciner derived reduced equations for the $n$-body problem using
 a method   similar to the method we   just used. 
See equation $(N Rel)$ in \cite{AlbouyChenciner98}.
  We now describe the equivalence between their reduced  equations and  our   reduced equations. 
  
Call  our Wintner-Conley matrix $A$.  Their  Wintner-Conley matrix is our $AM^{-1}$ which need not be symmetric.
Our formulation,   based as it is  on $sp(2n-2)$,  which is a  space
of symmetric matrices, requires  the  various matrices arising in the equations to be  symmetric.    
  Write our Gram matrix in terms of 4 symmetric blocks $b,c,d,r$ as
\[
 G =   \begin{pmatrix} b & c + r\\
 c - r & d \end{pmatrix} \,.
\]
Rewriting the Lax equation  \eqref{Lax} in this notation 
gives
\begin{equation}
\dot G =  \begin{pmatrix} 
  (c+r)\tilde M^{-1} + \tilde M^{-1} ( c - r)                   &  \tilde M^{-1} d -  b \tilde A  \\ 
  d \tilde M^{-1}  -  \tilde A  b  & -\tilde A  (c-r)  - (c+r) \tilde A 
\end{pmatrix} \,,
\label{eq:NRel}
\end{equation}
Choose  normalized Jacobi vectors so that $\tilde M = I$.  Then  the top left block reduces to $2c$,
and the entire equation becomes   their equation $(N Rel)$ up to the factor $-2$ in $A$.  See \cite{Chenciner12} and \cite{Moeckel} 
for more details.

Here is an alternative derivation of the equivalence of the two reduced  equations.   
Define $\hat q = q M^{1/2}$,  $\hat p =  p M^{-1/2} = \dot q M^{1/2} $ and $\hat Z = ( \hat q, \hat p)$ and hence $\hat G = \hat Z^t \hat Z$ and 
$\hat K = J \hat G$.  
In the new symplectic variables Newton's equations written as a first order system are $\dot {\hat q} = \hat p$, $\dot{\hat p} = \hat q \hat A$ 
where $\hat A = M^{-1/2} A M^{-1/2}$ which is symmetric.   Impose
the constraints of zero center of mass and
zero linear momentum: $\sum m_a q_a = \sum \sqrt{m_a} \hat q_a = \hat q w = 0$
where $w = ( \sqrt{m_1}, \dots, \sqrt{m_n} )^t$, and similarly $\sum p_a = \sum m_a v_a = \hat p w = 0$
so that both $\hat Z \in P_{d, n}$ and $\hat G$ have the kernel $(w^t,w^t)^t$. 
Compute 
\[
    \dot{ \hat K} = [ \hat P, \hat K], \quad \hat P = J \hat S, \quad \hat S = \begin{pmatrix} I & 0 \\ 0 & \hat A \end{pmatrix}\,. 
\]
Splitting $\hat G$ into symmetric blocks $B, C, D,  R$ now literally gives the equations of Albouy and Chenciner \cite{AlbouyChenciner98},
equations \eqref{eq:NRel} with $\tilde M = I$ and  lower case $b,c,d,r$ replaced by upper case $B,C,D,R$, 
in particular also $\dot R = [ \hat A, B]$.
This `hat'   Lax pair equation for evolution in $sp(2n)$  leaves the matrices
with kernel  $(w^t,w^t)^t$ invariant.  This  subalgebra of $sp(2n)$
can be identified with  $sp(2n-2)$ if desired.   

 \begin{remark}The observation that the invariants can be thought of as comprising
the subalgebra of  $sp(2n)$   with a fixed kernel 
was made in \cite{Dullin13}, where in a different basis the matrix $\hat G$ was represented as symmetric $2\times 2$ block-Laplacian matrices. 
In this way, by using a slightly modified basis of invariants $\hat G$, we can retain the simplicity of the Albouy-Chenciner equations without extra factors involving the mass matrix, but also retain the symplectic algebra in its standard basis.
A similar ``hat-construction'' can be done   by defining $\hat X = X \tilde M^{1/2}$,
$\hat Y = Y \tilde M^{-1/2}$ etc., where $X, Y$ are our Jacobi-vector based $d \times (n-1)$ matrices,
using the  translation reduced symmetric Wintner-Conley matrix $\hat A = \tilde M^{-1/2} \tilde A \tilde M^{-1/2}$.
   \end{remark}

    \begin{remark}
        One  difference in our  two approaches to reduction 
 is that Albouy-Chenciner  use  the language of dispositions  while we do not.
 Dispositions allowed them  to  construct 
 the quotient space by translations and boosts without choosing a basis for the label space $\R^n$
 and thus allow them to avoid   Jacobi vectors.  
 Another difference  is that their matrix $Z$ contains velocities  while ours is based on
  momenta, and hence
the basis of quadratic invariants looks  different  in our two representations. 
\end{remark} 

 \section{Proving the spatial symplectic reduction theorem}
 \label{s: spatial reduction}
 
 Here we  prove  theorem  \ref{thm: spatial reduced space}.
What remains to prove is the structure of the singularities  described in the last two sentences
of the theorem. 
Recall what we already know.  Our initial coadjoint orbit
 consists of rank 3 matrices  having the single spectral  invariant $\omega^2 = \|L \|^2$.
 These matrices  can all be written as   $G(Z)$ where $Z \in P_{3, n-1}$ has rank 3 and $\|L(Z) \|^2 = \omega^2$.
 The closure of this orbit contains exactly one new orbit, which consists of rank  2 matrices having  the same spectral invariant. 
  
  We may assume, by $O(3)$ equivariance of the momentum map,  that the angular momentum $\mu = L(Z)$ at which we are 
  performing reduction  points along the positive
  $z$-axis.     Identify  {\it planar}  centered phase space $P_{2, n-1}$
   with the linear   subspace  of $P_{3, n-1}$ consisting of matrices whose column vectors
   are perpindicular to $\mu$.  In other words, $P_{2, n-1}$ consists of  the $3 \times 2n-2$ matrices
   whose third row is identically zero.  
   Clearly we can find planar phase points $Z' \in P_{2, n-1}$ having  angular  momentum $\mu$.
        Since we can
   achieve any planar  phase point in $P_{3, n-1}$ as a limit of spatial ones, the image
   $G' = Z ^{' t} Z'$ of such a planar point  
   lies in the closure of our rank 3 coadjoint orbit.  It follows that the points added
   by taking closure of the rank 3 orbit  consists of rank 2 matrices
   sharing the same   spectral invariant.

{\it Rank one points do not lie on our  level set.}  We have assumed $\mu \ne 0$.
If  $L(Z) \ne 0$ then   $Z$ cannot have rank $1$.  For  if $Z$ had rank $1$ 
we would have all its columns  
$X_i$ and $Y_i$   proportional to the same unit vector
which would imply that $X_i \wedge Y_i = 0$ and thus   $L(Z) = \sum_i X_i \wedge Y_i  = 0$.   In other words, collinear configurations have angular momentum zero.

We now turn to singularities of our reduced space.
These singularities can arise  in two ways: as singular points of the level set $L^{-1}(\mu)$,
or as singularities of the quotienting operation by $G_{\mu}$.  We will
show that  our singularities all arise in the second way.  In the process we will show
that the singular points correspond to the rank 2 points.

 \begin{lemma}
 (Particular to  $d =3$ and $d=2$.) The  level set $L^{-1} (\mu)$ is smooth whenever $\mu \ne 0$.
 \end{lemma}
 This lemma follows immediately from the 
 fact, just proven, that if $\mu \ne 0$ then $L^{-1}(\mu)$ has no rank 1 or 0 points,
 and the following lemma.  
    \begin{lemma} If $Z \in P_{3,d}$ has rank 2 or 3 then
    the differential $dL(Z): P_{3,d} \to \R^3$ 
    is onto.     \label{l: rank lemma}
   \end{lemma}
   
 Lemma \ref{l: rank lemma} follows from a  general fact concerning
   Hamiltonian actions.   The    context is that
   of a compact Lie group $\G$  acting in a Hamiltonian fashion on
   a smooth symplectic manifold $P$ with momentum map $L: P \to Lie(\G)^*$.
   The $\G$-action is said to be ``locally free'' at $Z \in P$ if the isotropy
   subgroup $\G_Z$ at $Z$ is finite.  
   \begin{lemma}   $dL(Z): T_Z P \to Lie(G)^*$ is onto   if and only if 
   the $\G$   action is locally free at $Z$. 
   \label{l: loc free} 
   \end{lemma}
     This fact is well-known
   in this context and can be found in textbooks.   We recall its proof for completeness. 
 \begin{proof} (of lemma \ref{l: loc free})
  To say   that the $\G$ action is locally free at $Z \in M$
  means that the $\G$ orbit, $\G Z$, through $Z$  has dimension $\dim(\G)$.
  But the  momentum map generates the action, infinitesimally. It follows that   the
  Hamiltonian vector fields $X_{L_i}$   of the components $L_i, i =1, \ldots, dim(\G)$  of our momentum
  map span $T_Z (\G Z )$ when these fields are evaluated at $Z$.  But the  `symplectic gradient map'  $J_Z : T^* _Z P \to T_Z P$ 
  taking one-forms  like  $dL_i (Z)$ to their Hamiltonian vector fields $X_{L_i} (Z) = J_Z (dL_i (Z))$ is an invertible linear map.
  (It is the  ``inverse'' of the symplectic form).
  It follows that   the span of the 
   differentials $dL_i (Z)$ has the same dimension as that of the orbit $\G Z$.  
   Consequently, the action at $Z$  is locally free if and only if $\rank(dL(Z)) = \dim(\G)$   i.e.
   if and only if $Z$ is a regular value of $L$. 
   \end{proof} 
      
    We are now in a position to prove the rank lemma, lemma
    \ref{l: rank lemma}  as a corollary of lemma \ref{l: loc free}.
   
   {\it Spatial points are locally free.} 
   Suppose that $Z \in P_{3, n-1}$ has  rank $3$. We claim that  the isotropy of $Z$ is the identity
   group.   
  Since $Z$ has rank 3   we can choose three column vectors
   from the columns of   $Z$    which form a basis for $\R^3$.  If $g \in O(3)$
   fixes $Z$ then it  fixes this basis and hence is the identity on $\R^3$.    
   
  {\it Planar points are locally free.}  Next suppose that  $Z \in P_{3, n-1} $ has rank 2.  Then its   column vectors span a plane.
  Select two such vectors which form a basis for this plane.   
Any $g \in O(3)$ which stabilizes $Z$ must stabilize these two vectors and hence  be the identity on this plane.
The only non-trivial possibility for $g \in O(3)$ with $gZ = Z$ is either the identity or   reflection about this plane. 
The isotropy of $Z$ is this two-element reflection group $\Z_2$.    

QED

 To summarize:   $L^{-1}(\mu)$
 is smooth and consists of rank 3 and rank 2 matrices $Z$.
 The former have isotropy group the identity.
 The latter have isotropy group $\Z_2$ the two-element
 group of reflections about the plane $z = 0$.
 
 \begin{remark}
 \label{1 on so3}
 The momentum maps $L$  for the $O(3)$ action and for the $SO(3)$
 action are the same map, so that $L^{-1}(\mu)$  continues to be
 smooth for $SO(3)$.  What changes when we go to $SO(3)$
 is that now the action is free everywhere: planar points also have trivial isotropy.
 \end{remark} 
 
Continuing the reduction process, next we are to form $L^{-1}(\mu)/G_{\mu}$ where $G_{\mu} \subset O(3)$
 is the subgroup which leaves $\mu$ fixed under conjugation.
 
 {\it Observe   $G_{\mu} = O(2)$.} 
We have taken  $\mu =  (0,0, \|L\|)$  
 aligned  with the $z$-axis.
The isotropy subgroup of $\mu \in o(3)$  under the $O(3)$
action  then consists of the subgroup $O(2) \subset O(3)$
which preserves the splitting  of $\R^3$  into $\R^2 \oplus \R$.   (This isotropy fact  can be checked most
easily using the matrix realization of angular momenta.   Then the  group  $O(3)$ acts by conjugation on $\mu$'s.)

   Singularities of a quotient arise exactly
  where the isotropy type jumps. In our case jumps happen   at  planar points.
  We have proven the penultimate sentence of theorem  \ref{thm: spatial reduced space}
  
  We are now at the point of validating  the  last sentence of the theorem, the structure of the singularity.
 To achieve this we use some of  the basic theory of   smooth actions of compact Lie groups on smooth manifolds.
 For a synopsis of  this theory we recommend   the first chapter of Hsiang \cite{Hsiang}. (See in particular Theorem (I.5) on p. 11,
 the differentiable slice theorem.)
  Let $\G$, a compact Lie group, act smoothly on a connected  manifold $P$.  Suppose that
  for some points the isotropy group is trivial: $\G_Z = Id$.  Then the set of all such
  points $Z$  with $\G_Z = Id$  is open and dense and their projected image to the
  quotient space   $P/\G$ form the smooth points of this quotient.
  So suppose $Z$ is a point with $\G_Z \ne Id$.   We want
  to understand the local structure of the quotient near $\pi(Z) \in P/\G$
  where $\pi$ is the quotient map. For this it is
  enough to know  the {\it isotropy representation} at $Z$.  
  We recall how this works.  Since $\G_Z$ fixes $Z$ it acts linearly on $T_Z P$. 
  This linear representation is called {\it the isotropy representation}.    Put a $\G$-invariant
  \Ri{} metric on $P$.  Then, since $\G_Z$ maps the orbit $\G Z$ through $Z$
  to itself, the isotropy representation splits as
  $T_Z P = T_Z (\G Z) \oplus (T_Z (\G Z))^{\perp}$.
  Set $N_Z = (T_Z (\G Z))^{\perp}$, the normal space to the orbit. 
  We are   interested in the restriction of the $\G_Z$ representation
  to $N_Z$.  Since $\G_Z$ is compact (finite in our case!)
  we have that $N_Z/\G_Z$ is a nice topological space.  
  The basic theory, as summarized by this differentiable slice theorem, tells us that a \nbhd of $\pi (Z)$ in the quotient space 
  is diffeomorphic (as a singular variety) to the space  $N_Z/ \G Z$.
  
  How does the isotropy representation  play out for us?  We have $\G = O(2) \subset O(3)$,  
  $P =  L^{-1} (\mu) \subset P_{3, n-1}$ and that the only points with non-trivial
  isotropy are the planar points   $Z \in L^{-1}(\mu) \cap P_{2, n-1}$.  
  For these planar points we have that    $\G_Z = \Z_2$.  
  Any representation space $\V$ for $\Z_2$ splits into $\V = \V_+ \oplus \V_-$
  where the nontrivial element $\sigma$ of $\Z_2$ acts as the identity on $\V_+$
  and $\sigma$ acts as  minus the identity,  on $V_-$.  Then $\V/\Z_2 = \V_+ \times  (\V_- / \Z_2)$
  where $\V_- /\Z_2$ is the quotient space of $\V_-$ by the equivalence relation $v \sim -v$ for  $v \in \V_-$.
  We can understand this last quotient by restricting the quotient map  to the unit sphere in $\V_-$.
  The quotient space of the sphere by this same equivalence relation is   $\R \bP^\ell$ where $\ell = dim(\V) -1$ is the dimension of the sphere.
  Consequently $(\V_- / \Z_2 )  \cong Cone(\R \bP ^\ell)$ and $\V/\Z_2 = \R^s \times  Cone(\R \bP ^{\ell})$
  where $s = dim(\V_+),  \ell+1 = dim(\V_-)$.
  
  In our situation it remains only to find the dimensions $s, \ell$
  for the $\pm$ decomposition of the isotropy representation of $\Z_2$ 
  on $N_Z = T_Z (\G Z) ^{\perp} \cap \ker( dL(Z))$ where $L(Z) = \mu = \omega e_3$
  and $Z \in P_{2, n-1}$.  
  To this end, we first decompose all of $P_{3, n-1}$ and then
  restrict the decomposition to  $N_Z \subset  \ker ( dL(Z))$.   
The nontrivial element $\sigma \in \Z_2$ acts on column vectors
  by 
  $$\sigma \begin{pmatrix} 
x  \\
y  \\
   z  
  \end{pmatrix}   = 
  \begin{pmatrix} 
x  \\
y  \\
 -z  
  \end{pmatrix}  .
  $$
  It follows that the $\Z_2$-decomposition of $P_{3,n-1}$ is the
  obvious horizontal - vertical decomposition: 
  $$P_{3, n-1} = \V_+ \oplus \V_- := P_{2, n-1} \oplus \nu$$
  where  
   $$ Z = \begin{pmatrix} X_{1, 1}  & X_{2, 1} & \ldots  & Y_{n-1, 1} \\
   X_{1, 2}  & X_{2, 2} & \ldots  & Y_{n-1, 2}  \\
0 & 0  & \ldots  & 0    \end{pmatrix}  \in P_{2, n-1} $$
while matrices   in 
 $\nu$ only have nonzero entries in their last row: 
$$ \eta =   \begin{pmatrix} 
0 & 0  & \ldots  & 0 & \ldots & 0  \\
0 & 0  & \ldots  & 0 & \ldots & 0  \\
   z_1   & z_2 & \ldots  & w_1 & \ldots &   w_{n-1} 
  \end{pmatrix}   \in \nu . $$
  The dimension of $\nu$ is $2n-2$. 
  If $Z \in P_{2, n-1}$ is as above with   $J(Z) = \mu = \omega^2 e_3$,
  and if $\eta  \in \nu$ as above, then we compute
  $J(Z + \eta) = \mu + ((\sum w_a X_a - z_a Y_a) \times e_3)$.
  Since the column vectors $X_a, Y_a$ of $Z$  lie in the $\R^2$ perpendicular to $e_3$  
  we see that $J(Z + \eta) = \mu \iff \sum w_a X_a- z_a Y_a = 0 \in \R^2$.
  Since $Z$ has rank 2 this last equation, viewed as a linear equation in $w$ and $z$
  is two linearly independent equations. Thus  $L^{-1}(\mu) \cap (\{Z \} + \nu )$
  is a codimension $2$ linear space of $\nu$.
  This linear space is the minus part in the $\Z_2$ decomposition of    $\ker(dL(Z))$.
  It follows that the dimension of this minus part is  $\ell +1 = 2n-4$.
  
  Finally, we verify that the dimension of the plus part comes out as it must.
  The orbit $\G Z$ is a circle lying in  $L^{-1}(\mu) \cap P_{2,n-1}$
  which in turn has codimension one within $P_{2, n-1}$.
  It follows that   the remainder of $N_Z$, namely $T_Z (\G Z) ^{\perp} \cap \ker(dL(Z)) \cap P_{2,n-1}$
  has codimension two within $P_{2,n-1}$ and hence
  dimesion $s = 4n- 6$.  
  
  QED
  
  \begin{remark}
  \label{2 for so3}
  When we look at the $SO(3)$-reduced space we find that $G_{\mu} = SO(2)$
  but more importantly that action is free everywhere, as mentioned in remark \ref{1 on so3}.
  It follows that the reduced space $L^{-1}(\mu)/SO(2)$ is an everywhere smooth symplectic manifold. 
  \end{remark}

\appendix
 

  \section{Symplectic normal forms}       
 \label{s: symp norm form} 
  
In this appendix we derive the normal form
\eqref{normal form}   for positive semi-definite elements  of $sp(2n-2)$.
We do this by using the  Singular-Value-Decomposition-like factorisation due to Xu \cite{Xu}.
Xu showed that every matrix $Z \in P_{d, 2n-2}$ can be written: 
\[
   Z = Q D T^{-1},
\]
with  $Q \in SO(d)$, $T \in Sp(2n-2)$ and $D \in P_{d, 2n-2}$ a  `permuted diagonal'
of  the following shape: 
\[
D = 
\begin{array}{c@{\;}c@{\;}c@{\;}c@{\;}c@{\;}c@{\;}c}
  &
  \begin{array}{c} p \\ q \\ p \\d '  \end{array} &
  \left(
       \begin{array}{ccc|ccc}
    \Sigma  &  0 & 0 & 0 & 0 & 0 \\
    0 & I & 0 & 0 & 0 & 0   \\
    0  & 0 & 0 & \Sigma & 0 & 0   \\
     0 & 0 & 0 & 0 & 0 & 0 
  \end{array} 
 \right)
\end{array}
 \qquad \text{ with } \Sigma > 0 \text{ diagonal } \]
The  nonzero entries  of $D$ are square matrices of the indicated size,  
 $\Sigma$ is  positive definite diagonal and $I$ the identity.  We set 
$d' = d - (p+q)$.  The   3rd and 6th column  have  
width $n' = n-1 - (p+q)$ and may be absent.  The $2-3$ and $4-6$ block entries  in $D$ which consists of $0$'s are
typically not square matrices.
$D$ encodes all the spectral and normal form information within
  $ L(Z)$ and $G(Z)$. Indeed 
  \[
   L = Z J Z^t  \sim_{O}  \; DJD^t   \in so(d)
\]
and
\[
   G =  Z^t Z \sim_{Sp} \;  D^t D   \in sp(2n-2)\,.
\]
where the symbols `$\sim_{O} $ and `$\sim_{Sp} $' mean that the indicated matrices
are conjugate via  elements from   $O(d)$ and $Sp(2n-2)$.  When $G$ is turned into a   quadratic Hamiltonian then Xu's  normal form yields    
\[
   \lambda = \frac{1}{2} \sum_{j=1}^p \omega_j ^2 (x_j^2 + y_j^2) + \frac{1}{2} \sum_{p < j \le p+q }^q y_j ^2  
   \]
where the $\omega_j ^2$ are the   diagonal entries  of $\Sigma$.  We have
found that this quadratic Hamiltonian representation just given    is the most efficient
way to read out the   information contained in elements of $sp(2n-2)$.

A reordering of  coordinates puts   $DJD^t  \sim L \in so(d)$   into  block-diagonal form  with $p$ non-zero
$2\times 2$ blocks $\begin{pmatrix}
0 & \omega_j^2 \\ -\omega_j^2 & 0 
\end{pmatrix}$.  Precisely the same blocks appear in $J D^t D \sim K$.
These blocks  represent the nonzero   eigenvalue pairs $\pm i \omega_j^2$ of both matrices. 
The remaining blocks of $DJ D^t$
are zeros, and represent  the  eigenvalue
$0$ with multiplicity $d - 2p$.  But 
  $JD^tD \sim K \in sp(2n-2)$ has $q$ Jordan blocks representing generalized $0$ eigenvectors
  which arise as  
  $2\times 2$ blocks 
  of the form  $\begin{pmatrix}
0 & 1\\ 0  & 0 
\end{pmatrix}$.  All   remaining  $n-1-p-q$ diagonal blocks of $JD^t D$  of size $2\times2$  are  zero.

Xu's ``dual pair'' factorisation provides another proof of Theorem~\ref{spectral lemma}, while providing  more detailed information 
about the rank of  $G$. 
\begin{lemma}
Let $p$ be the number of non-zero eigenvalue pairs in $L(Z)$.
Then $p$ is also the number of non-zero eigenvalue pairs in $JG(Z)$.
Let   $q$ be the number of $2\times 2$ Jordan blocks with eigenvalue 0 in $JG(Z)$.
Then the  rank $r$ of  $JG(Z)$ is  $2p + q$ while the rank of $L$ is  $2p.$
The integers $(p,q)$ satisfy the   inequalities  $p+q \le n-1$ $q \le \min( d - 2p, n-1-p)$,
  $p \le d/2$ and   $2p +q  \le min(d, 2n-2)$.  
\end{lemma}

\label{sec:appendixB}

\section{Typical coadjoint orbits}

In this appendix we describe some basic general  properties
of the   coadjoint orbits   associated to the $n$-body problem.
 For the purposes of this appendix and the next  one set 
 $$m = n-1 , \qquad G = Sp(2m, \R), \text{ and } \qquad   \g = sp(2m) \cong sp(2m)^*.$$
 We   identify $\g$ with the vector space of quadratic Hamiltonians $h = h(x_1, \ldots, x_m, y_1, \ldots, y_m)$, 
 with its Lie bracket being their  Poisson bracket.
 When needed we will write $Gram(Z) = Z^t Z$ for 
 the Gram map, which is not to be confused in this appendix and the next, with
 our Lie group.

Write $\mathcal O_{\lambda} = \{ g^* \lambda:  g \in G \}$
for the coadjoint orbit of the quadratic Hamiltonian $\lambda \in \g$.
From general theory $\mathcal O_{\lambda}$   is diffeomorphic, as a $G$-space, to the    homogeneous space $G / G_\lambda$
where $G_{\lambda} = \{g : g^* \lambda = \lambda \}$ is the isotropy subgroup of $\lambda$.
Write
$$\g_{\lambda} = \{ f \in sp(2m):  \{ f,  \lambda \} = 0 \}$$
 for the Lie algebra of $G_{\lambda}$.  The isotropy group $G_{\lambda}$ need not   be connected.
 Write $(G_{\lambda})^0$ for the identity component of $G_{\lambda}$.
 Then   $(G_{\lambda})^0$ is uniquely determined from $\g_{\lambda}$ as the
 connected Lie subgroup of $G$ having Lie algebra $\g_{\lambda}$. 
 The quotient group  
 $G_{\lambda}/ G_{\lambda}^0$ is a finite group.  We will not pretend to understand
 this finite group or the full orbit $G/G_{\lambda}$.
 Rather we will content ourselves with computing
 $dim(\mathcal O_{\lambda})$ and $G_{\lambda} ^0$, which is the same as knowing the orbit
 up to finite cover.  
 Since $dim(G_{\lambda}) = dim(\g_{\lambda})$,    $dim(\mathcal O_{\lambda}) = dim(G) - dim(\g_{\lambda})$ and  $dim(G) = {2m+1 \choose 2} = m(2m+1)$  we get
$$dim(\mathcal O_{\lambda}) = m (2m + 1) - dim(\g_{\lambda}).$$
We have   $dim(\g_{\lambda}) \ge m$
and that $dim(\g_{\lambda}) =m$ for an open dense set of $\lambda$'s.
We call these $\lambda$'s and their orbits ``generic''. 
 Thus the generic orbit has
dimension $2m^2$.    Note that $m$ is also the number of Casimirs.
These generic orbits are characterized by having numerical invariants $(p,q) = (m,0)$
or $(m-1, 1)$ and all   spectral invariants distinct from each other. 
Of these, those  orbits having $p=m$ are closed and are uniquely characterized as
being   level sets of the Casimirs.   Orbits having  $(p,q) = (m-1, 1)$
have orbits with $(p,q) = (m-1, 0)$ in their closure.

Since our interest is in the $n$-body problem we are only interested in 
the orbits lying in $symm_+(2n-2)$, the   subset of  positive semi-definite quadratic Hamiltonians.
See theorem \ref{Poiss reduction}.  
 Recall definition \ref{lem: invariants} and  lemma  \ref{lem: invariants}:   the orbit type of   $\lambda \in symm_+$ 
 is completely determined by
 its  numerical invariants $(p, q)$ and  its $p$ spectral invariants $\omega_j ^2, j =1, \ldots, p$ of $\lambda$ as given by the normal form.
 A  $\lambda$  having these invariants can be brought  into the normal form:
 \beq \lambda = \frac{1}{2} \sum_{j =1} ^p \omega_j ^2 (x_j^2 + y_j ^2) +   \frac{1}{2} \sum_{i = p+1} ^{p+q} x_i ^2.
\Leq{}
 See also equation  \eqref{normal form}.   When $q = 0$ of $(p, q)$ then  the second term is not present.  
 We have that  $p+q \le m= n-1$. 
 The motion space for the bodies of an $n$-body problem lying on $\lambda$'s  orbit
 is $d = 2p +q$ which is the rank of $K$.  The  orbit is closed if and only if $q = 0$.
 The values of the $m$ Casimirs $tr(K^{2\ell})$, $\ell =1, \ldots, m$ on $\lambda$  can be written in terms of its spectral invariants:
 $$tr(K^{2 \ell}) = (-1)^{\ell} \sum_{j=1}^p (\omega_j) ^{2 \ell},  \ell =1, 2, \ldots, m. $$ 
  Our goal then is to describe $dim({\mathcal O}_{\lambda})$ and $G_{\lambda}^0$
  in terms of the   numerical invariants $(p, q)$ and the spectral invariants  $\omega_j ^2$ of  $\lambda \in symm_+(2m)$.

 \begin{case} Generic case: 
 $p = m$ with  the $\omega_j ^2$ distinct. 
 These orbits have maximal possible dimension $2m^2$. 
 Since $q = 0$ they are closed and hence defined by setting the Casimirs to constants. 
 Write $I_j = \frac{1}{2} \omega_j ^2 (x_j ^2 + y_j ^2)$ so that
  $\lambda = \sum_{j =1} ^m  I_j$.  The   $I_j$ form a basis for $\g_{\lambda}$ yielding 
$dim(\g_{\lambda}) = m$ and 
 $dim(\Omu) = 2m ^2$ as claimed. 
 Each $I_j$ generates a circle $\mathbb{S}^1 = U(1)$ which rotates just the   $x_j -y_j$ plane.
 Thus $G_{\lambda} ^0$ is an m-torus, being the  $m$-fold product of circles $U(1)$.  
 This m-dimensional torus $\mathbb{T}^m$ is the maximal torus inside $U(m) \subset Sp(2m)$. 
 This $U(m)$ in turn is the maximal compact subgroup of $Sp(2m)$.  
 The dimension of $\mathbb{T}^m$ is $m$ so  $dim(\Omu) = 2m^2$. 
 \end{case} 
 
 \begin{case} Special case of the generic case: 
 $p = m-1$, $q = 1$,  with  the $\omega_j ^2$ distinct. 
 The orbit is again of maximal dimension $2m^2$ but it is not closed.
 Its closure is the orbit with $(p,q) =(m-1, 0)$ and the same spectral invariants.
Using the same notation as above, we   have the normal form  $\lambda = \sum_{j =1} ^{m-1} I_j + \frac{1}{2} x_m ^2$.
The isotropy algebra again has dimension $m$, being generated by the $I_j$ and $x_m ^2$.
$G_{\lambda} ^0$ is again commutative but now is non-compact, being of  the form of $\mathbb{T}^{m-1} \times \R$,
since $x_m ^2$ generates a shear in the $x_my_m$-plane. The dimension is the same as in the previous case.
\end{case}

 \begin{case} $p=m$ with all of the $\omega_j$ equal.
 The element  $\lambda = \frac{1}{2} \sum (x_j ^2 + y_j ^2)$
 is the momentum map for a $U(m)$
 action on $\R^{2m}$.  To describe this action form  $x_j + \sqrt{-1} y_j = z_j$ and thereby identify
 $\R^{2m}$ with $\C^m$.  Then the action is scalar multiplication by a unit complex number.
 We can think of $\C^m = \C^m \otimes \C^1$  
 in which case one sees this
 $U(1)$ as half of the Howe dual pair $(U(1), U(m))$.  It follows that
 $G_{\lambda} =  U(m)$.  The dimension of $U(m)$ is  $m^2$
leading to $dim(\Omu) = m(m+1)$.  \end{case}

\begin{remark}
The isotropy group $U(m)$ of  case 3 is   the
isotropy group of the almost complex structure $J$ we have been using,  where $J^2 = -I$.
( This   fact is closely related to the equalities  $SO(2m) \cap Sp(2m) = Gl(m, \C) \cap Sp(2m) = U(m)$.)
 It follows   that the coadjoint orbit of example 3 is    the space of almost complex structures
compatible with our symplectic form on $\R^{2m}$.  What is the  differential geometric
importance,  if any, of letting an almost complex structure evolve ``as if it were''
a reduced point in an $n$-body phase space? We don't know. 
\end{remark} 

 \begin{case} $p = 0$ and $q =m$. Then $\lambda = \frac{1}{2} \sum_{j=1} ^m  y_j ^2$
 is the Hamiltonian for a free particle moving in $\R^m$.  Think of $\lambda$
 as a full rank quadratic form on the $\R^m$ with coordinates $y_i$.  Then
 the associated group  $O(m)$ leaves $\lambda$  invariant and 
 embeds in $sp(2m)$ diagonally, acting in the same way on both $x_i$ and $y_i$. 
 This group has momentum map the usual angular momentum   with components $f_{ij} = x_i y_j - y_i x_j, 1 \le i, j \le m$.
 Of course  $\{\lambda, f_{ij} \} = 0$.  In addition to $f_{ij} \in \g_{\lambda}$,  any    polynomial in the momenta $y_i$
 commutes with $\lambda$.  Among these polynomials in $y_j$  the  
 quadratic ones form the   vector space  $symm(m)$ of quadratic forms on $\R^m$.
 The $f_{ij}$ do not Poisson commute with the elements of $symm(m)$  
  and a moment's thought reveals that their Poisson bracket relations   
   arise out of  the action of $SO(m)$ on $symm(m)$.
  To summarize then,  the connected component of $G_{\lambda}$ is 
 $G_{\lambda} ^{0} = SO(m) \ltimes symm(m)$ where 
the  semi-direct product arises from the action of $SO(m)$ on $symm(m)$. 
 \end{case}

 \begin{case} $p = 0$ but $0 < q < m$. Then $\lambda = \frac{1}{2} \sum_{j=1} ^q  y_j ^2$
 is the Hamiltonian for a free particle moving in $\R^q \subset \R^m$ with all the other variables $x_{\mu}, y_{\mu},  \mu > q$
 irrelevant.     Split  up the canonical  coordinates of $\R^{2m}$ into $x_i, y_i, 1 \le i \le q$
 and the complementary set $x_{\mu}, y_{\mu},  q < \mu \le m$.  The binomials in $x_i, y_i$ commuting  with $\lambda$
 fit together precisely as in the previous case to form the Lie algebra $so(q) \ltimes symm(q)$.
 Every  binomial  in $x_{\mu}, y_{\mu}$    commutes with $\lambda$ and together these span   the Lie algebra $sp(2(m-q))$.
 Finally we have mixed terms.  Any quadratic polynomial of the form $y_i x_{\mu}$ or $y_i y_{\mu}$ Poisson commutes with $\lambda$.
 These mixed polynomials form the symplectic 
  vector space $P_{q, m-q} = \R^q \otimes \R^{2(m-q)}$. We have now exhausted the commutator algebra $\g_{\lambda}$.
 Putting  these pieces together we see that $\g_{\lambda}$ is  the Lie algebra of the group
 $(SO(q) \times Sp(2m-2q) )\ltimes (symm(q) \times P_{q, m-q})$ which is  the identity component of the isotropy group of $\lambda$.    Here the semi-direct product structure
 $\ltimes$ is given by having  
 $(O(q), Sp(2m-2q))$ act on   $P_{q, m-q}$ just 
 like our guiding dual pair $(O(d), Sp(2n-2))$ acts on $P_{d, n-1}$, only with   dimensions shifted
 to $(q, m-q)$  from $(d, n-1)$, and then restricting this action to   $SO(q) \times Sp(2m-2q)$.  The $SO(q)$ factor alone acts on $symm(q)$.
 \end{case}

\section{Coadjoint orbits in low dimensions}

\subsection{The 2-body problem}

When $n =2$   the relevant Lie-Poisson structure occurs on the
dual of the three-dimensional Lie algebra    $sp(2) = sl(2, \R)$.  
We use the Killing form to identify the Lie algebra with its dual.
  The Killing form is a non-degenerate symmetric quadratic
form having Lorentzian signature $(2,1)$.  The  squared (Lorentzian)
length  for this form equals the  single Casimir  function $C =X^2 + Y^2 -Z^2$ 
in appropriate linear coordinates on $sl(2, \R)$ and equals the   square of the angular momentum.    The 2-body Poisson reduced space
forms the closure  of the positive light cone.    The interior of the cone is foliated by the level sets
$C = const. > 0$ each of which is a symplectic leaf sitting inside
$\R^{2,1}$.  The points of the interior have $(p,q) = (1,0)$.
The boundary of the cone, minus the cone point $0$,  corresponds  to points with 
$(p, q) = (0,1)$ and also to the locus $C = 0$ where the   angular momentum  is zero.  These 
boundary points form a single orbit with normal form $\frac{1}{2} x^2$ and are where the 
reduced    
collinear 2-body motions take place.

\begin{figure}[h]
\scalebox{0.4}{\includegraphics{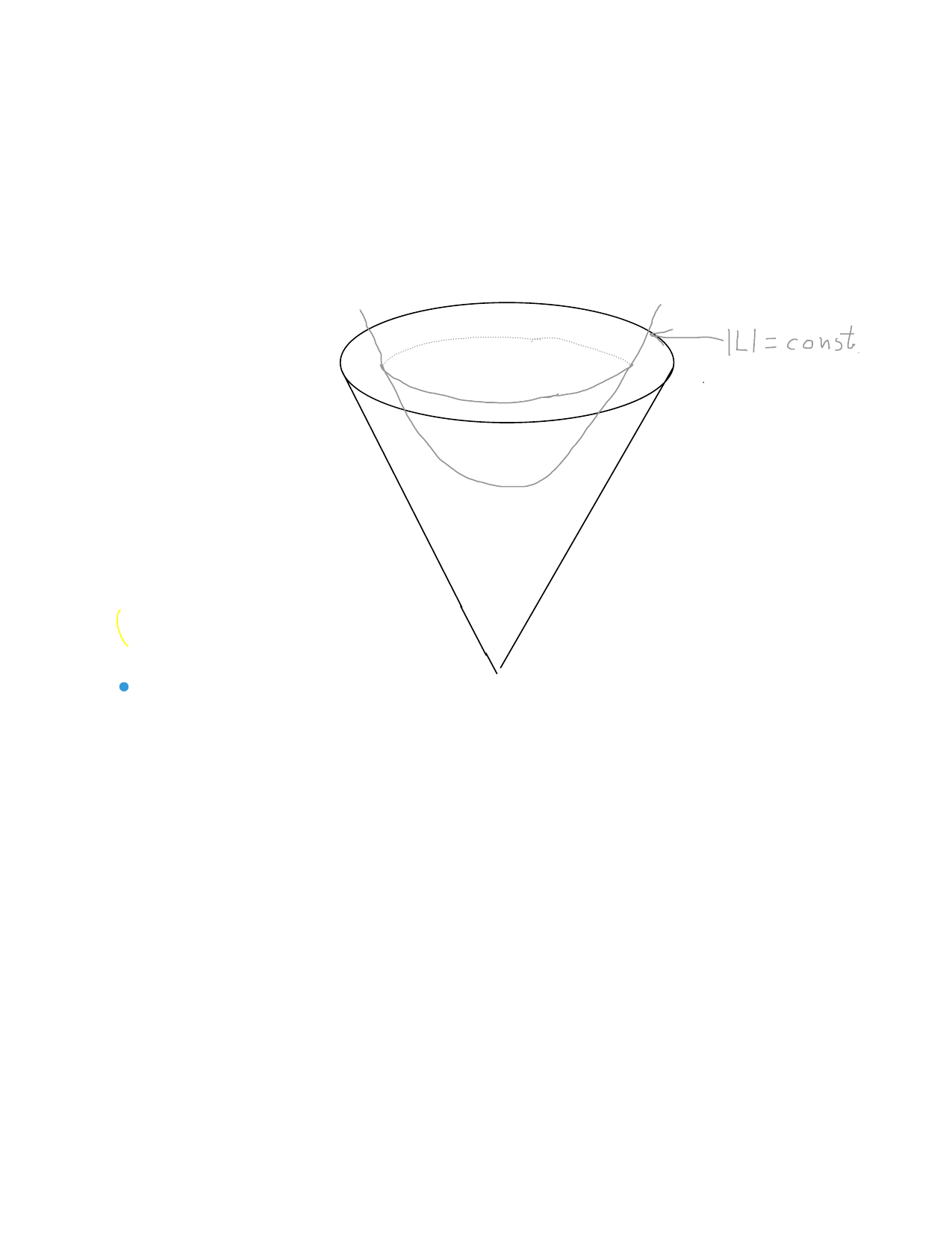}}
\caption{The $n=2$ Poisson structure with a symplectic leaf $|L| = const$.}  \label{cone1}
\end{figure}

\subsection{The 3-body problem}  The coadjoint orbits for $n =3$  live in the $10$ dimensional dual of the Lie algebra $sp(4)$.
The dimension
of the generic orbit is $8$. 
In the   table below  we   list all    positive semi-definite nonzero  coadjoint orbit types  arising when $n =3$.
The top row lists data for this  generic coadjoint orbit type.
 In each row we list an orbit type with its   dimension and the connected component of its isotropy group.   
In the bottom row $Heis_1$ denotes the 3-dimensional Heisenberg group.
This is the simply connected nilpotent Lie group with generators $X, Y$ and brackets $[X, Y] = Z$.
We will say a few words regarding how the Heisenberg group  arises.

{\bf Table for $n = 3$ below.}
 \[  
\begin{array}{ c|c|c|c|c}
  p  & q & (G_\lambda )^0   & dim (\Omu) & d\\ \hline
  2 \; (nondeg.)& 0 & SO(2)\times SO(2)  & 8 & 4 \\
       2  \; (deg.)   & 0 & U(2)   & 6 & 4 \\
    1 & 1 & SO(2) \times \R   & 8 & 3\\
    1 & 0 & SO(2) \times Sp(2)   & 6 & 2 \\
   0 & 2 & SO(2) \ltimes symm(2)   & 6 & 2\\
  0 & 1 & Sp(2) \ltimes (Heis_1)   & 4 & 1 \\ 
   \end{array}
\]

We explain the last row and in the process outline our methods of computation.
The normal form for this orbit is $\lambda = \frac{1}{2} y_1 ^2$.  It Poisson commutes with
itself and any function $f = f(x_2, y_2)$.  It also commutes with the functions of the
form $y_1 g (x_2, y_2)$.  The quadratic functions of the first type form the Lie algebra
$sp(2)$.  If we want functions of the second type to be quadratic  then $g$ must be linear.
 A basis for these functions of the second type is $X = y_1 x_2$ and $Y = y_1 y_2$.
We have $\{ X, Y \} = y_1 ^2$.  Set $Z = y_1 ^2$.  Replacing Poisson brackets with Lie brackets this means that we 
 $[X, Y] = Z$ while $[X, Z] = [Y, Z] = 0$ since $Z$ is  in the center of $\g_{\lambda}$.
 The Heisenberg algebra is the algebra with these commutation relations.
 Integrate up to get the Heisenberg group $Heis_1$.  
 
 \begin{remark}
 \label{rem:Heis}
 The computation just described generalizes without difficulty
 to the case $(p, q)  = (0,1)$ for any number $n$ of bodies.
 This orbit represents the zero angular momentum $n$-body problem.
 Its normal form is given by the   rank one element $\lambda = \frac{1}{2} y_1 ^2 \in sp(2m)$,
 The isotropy of $\lambda$ is  $Sp(2(m-1)) \ltimes Heis_{m-1}$ where $Heis_{m-1}$
 is the usual Heisenberg algebra of $\R^{2m -2} \oplus \R$
 associated to the symplectic form on $\R^{2m-2}$.
 It is interesting to note that this isotropy group  is precisely the group of automorphisms
 of this  Heisenberg algebra.
 
 \end{remark}

\subsection{The 4-body problem} 
The  Lie algebra for $n = 4$ is  $\g = sp(6)$ which has dimension $21$.  Its generic coadjoint  orbit has dimension $18$.
In the   table below  we list all orbits  for which  $d =2p +q \le 3$.  Recall that  this  $d$ is the dimension in which the bodies move
when we think of the orbit as being the principal stratum for the symplectic reduced space for $4$ bodies.   The generic $18$-dimensional
 coadjoint  orbits have  $d= 6$ and   corresponds to the reduction of the  4-body problem in 6 dimensions, reduced at a   generic rank 6 value of  angular momentum.
 We have ignored it in our table,  along with all other orbits having $d > 3$. 

{\bf Table for $n = 4$ below.}
  \[
\begin{array}{ c|c|c|c|c}
  p  & q & (G_\lambda )^0   & dim (\Omu) & d\\ \hline
      1 & 1 & SO(2) \times Sp(2) \ltimes \R^3   & 14 & 3\\
        1 & 0 & SO(2) \times Sp(4)   & 10 & 2 \\
    0 & 3 & SO(3) \ltimes symm(3)   & 12 & 3 \\
   0 & 2 & * \text{ see text below } * & 10 & 2\\
         0 & 1 &  Sp(2) \ltimes (Heis_2) & 6 & 1\\
 \end{array}
\]
The   $p = 0$, $q=2$ entry is a special case of 
`case 5' at the end of Appendix C.  Plugging in the
integers, we find that  the identity component of the isotropy group of this
$\lambda$   is
$(G_\lambda )^0 = H \ltimes (symm(2) \times P_{2,1})$. Here  $H = SO(2) \times Sp(2)$ and $P_{2,1} = \R^2 \otimes \R^2$.
In the semi-direct product $H$ acts on 
$P_{2,1}$    in our standard `dual pair' manner, $SO(2)$ acting on the first factor
$\R^2$ and $Sp(2)$ on the second $\R^2$ factor.     Only   $SO(2)$   acts on $symm(2)$. 

Regarding the $p = 0, q=1$ entry, see remark~\ref{rem:Heis} above.

\end{document}